\documentclass[11pt]{article}

\usepackage[letterpaper,margin=1in]{geometry}
\usepackage{amsmath, amssymb, amsthm, amsfonts}
\usepackage{bbm}
\usepackage{bm}
\usepackage{cite}
\usepackage{cleveref}
\usepackage{float}
\usepackage{subcaption} 
\usepackage{xspace} 
\usepackage{thm-restate} 
\usepackage{bm} 
\usepackage{ifthen} 
\usepackage{xstring} 

\usepackage[framemethod=tikz]{mdframed}
\usepackage[bottom]{footmisc}
\usepackage[inline]{enumitem} 
\setitemize{noitemsep,topsep=3pt,parsep=3pt,partopsep=3pt}

\usepackage{tikz}
\usetikzlibrary{arrows,shapes,trees}

\usepackage{catchfile}
\newcommand{\getenv}[2][]{%
  \CatchFileEdef{\temp}{"|kpsewhich --var-value #2"}{}%
  \StrGobbleRight{\temp}{1}[\newtemp]
  \if\relax\detokenize{#1}\relax\temp\else\edef#1{\newtemp}\fi%
}

\newcommand{\EnvFullOrShort}{full}     

\ifthenelse{\equal{\EnvFullOrShort}{short}}{
  \wlog{SHORT VERSION}
  
  \newcommand{\fullOnly}[1]{}
  \newcommand{\shortOnly}[1]{#1}
}{
  \wlog{FULL VERSION}
  \newcommand{\fullOnly}[1]{#1}
  \newcommand{\shortOnly}[1]{}
}

\newcommand{\cd}{\cdot}

\newcommand{\lds}{\ldots}

\newcommand{\bs}{\setminus}
\newcommand{\s}{\subseteq}

\newcommand{\BE}{\begin{enumerate}}
\newcommand{\EE}{\end{enumerate}}
\newcommand{\im}{\item}
\newcommand{\BI}{\begin{itemize}}
\newcommand{\EI}{\end{itemize}}

\newcommand{\Bigcup}{\displaystyle\bigcup\limits}

\newcommand{\logn}{\log n}

\newcommand{\R}{\mathbb R}

\newcommand{\N}{\mathbb N}

\newcommand{\e}{\epsilon}

\newcommand{\m}{\mathcal}
\newcommand{\mc}{\mathcal}

\newcommand{\tO}{\tilde{O}}

\newcommand{\lc}{\lceil}
\newcommand{\rc}{\rceil}
\newcommand{\lp}{\left(}
\newcommand{\rp}{\right)}

\newcommand{\lmt}{\left[\begin{matrix}}
\newcommand{\rmt}{\end{matrix}\right]}
\newtheorem{theorem}{Theorem}
\newtheorem{lemma}{Lemma}
\newtheorem{corollary}{Corollary}
\newtheorem{definition}{Definition}
\newtheorem{observation}{Observation}
\newtheorem{fact}{Fact}
\newcommand{\BT}{\begin{theorem}}
\newcommand{\ET}{\end{theorem}}
\newcommand{\BL}{\begin{lemma}}
\newcommand{\EL}{\end{lemma}}
\newcommand{\BC}{\begin{corollary}}
\newcommand{\EC}{\end{corollary}}
\newcommand{\BD}{\begin{definition}}
\newcommand{\ED}{\end{definition}}
\newcommand{\BO}{\begin{observation}}
\newcommand{\EO}{\end{observation}}
\newcommand{\BP}{\begin{proof}}
\newcommand{\EP}{\end{proof}}
\newcommand{\BPS}{\begin{proof}[Proof (Sketch)]}
\newcommand{\EPS}{\end{proof}}

\newcommand{\hdr}[1]{{\vspace{2mm}\noindent\textbf{#1:}}}

\newcommand{\dectree}{\mathbb{DT}} 

\newcommand{\desc}{\text{desc}}

\DeclareMathOperator{\poly}{poly}
\usepackage{graphicx}

\begin{document}

\renewcommand{\thefootnote}{\fnsymbol{footnote}}

\title{Minor Excluded Network Families Admit
  Fast Distributed Algorithms\footnotemark[1]}
\author{Bernhard Haeupler\footnotemark[2], Jason Li\footnotemark[2], Goran Zuzic\footnotemark[2]} 
\date{\today}

\footnotetext[1]{This work was supported in part by NSF grants CCF-1527110 "Distributed Algorithms for Near Planar Networks" and CCF-1618280 "Coding for Distributed Computing".}
\footnotetext[2]{Carnegie Mellon University, Pittsburgh PA, USA. E-mail: \{haeupler,jmli,gzuzic\}@cs.cmu.edu.}

\maketitle

\renewcommand{\thefootnote}{\arabic{footnote}}
\setcounter{footnote}{0}

\pagenumbering{gobble}

\begin{abstract}
Distributed network optimization algorithms, such as minimum spanning tree, minimum cut, and shortest path, are an active research area in distributed computing. This paper presents a fast distributed algorithm for such problems in the CONGEST model, on networks that exclude a fixed minor.

On general graphs, many optimization problems, including the ones mentioned above, require $\tilde\Omega(\sqrt n)$ rounds of communication in the CONGEST model, even if the network graph has a much smaller diameter. Naturally, the next step in algorithm design is to design efficient algorithms which bypass this lower bound on a restricted class of graphs. Currently, the only known method of doing so uses the low-congestion shortcut framework of Ghaffari and Haeupler [SODA'16]. Building off of their work, this paper proves that excluded minor graphs admit high-quality shortcuts, leading to an $\tilde O(D^2)$ round algorithm for the aforementioned problems, where $D$ is the diameter of the network graph.
To work with excluded minor graph families, we utilize the Graph Structure Theorem of Robertson and Seymour. To the best of our knowledge, this is the first time the Graph Structure Theorem has been used for an algorithmic result in the distributed setting.

Even though the proof is involved, merely showing the existence of good shortcuts is sufficient to obtain simple, efficient distributed algorithms. In particular, the shortcut framework can efficiently construct near-optimal shortcuts and then use them to solve the optimization problems. This, combined with the very general family of excluded minor graphs, which includes most other important graph classes, makes this result of significant interest.



\end{abstract}

\newpage
\pagenumbering{arabic}

\section{Introduction}
Network optimization problems in the CONGEST model, such as Minimum Spanning Tree, Min-Cut or Shortest Path are an active research area in theoretical distributed computing~\cite{garay1998sublinear,khan2008fast,elkin2017simple,nanongkai2014almost}. This paper provides a fast distributed algorithm for such problems in excluded minor graphs in the CONGEST model.

   Recently, lower bounds have been established on many distributed network optimization problems, including all of the mentioned ones~\cite{sarma2012distributed}. More specifically, each of these problems in the CONGEST model require $\tilde{\Omega}(\sqrt{n})$\footnote{Throughout this paper, $\tilde O(\cdot)$ and $\tilde \Omega(\cdot)$ hide polylogarithmic factors in $n$, the number of nodes in the network.} rounds of communication to solve. This holds even on graphs with small diameter, for example, when the diameter is logarithmic in the number of nodes $n$. The result is surprising since it is not immediately clear why a network of small diameter requires such a large number of communication rounds when solving an optimization problem. 


  On the positive side, the next major question in algorithmic design is to determine whether one can bypass this barrier by restricting the class of network graphs. One immediate question that arises is, what family of graphs should one consider? Ideally, such a class of graphs should be inclusive enough to admit most ``realistic'' networks, yet be restrictive enough to disallow the pathological lower bound instances. 

  In our search for a restricted graph family to study, we focus on three criteria. First, we desire a family with a \textit{rich and rigorous mathematical theory}, so that our result is technically meaningful. Second, the family should capture many networks in practice. And finally, we want \textit{robustness}: a graph with a few added or perturbed edges and vertices should still remain in the family. Robustness is an important goal, since we want our graph family to be resistant to noise. For example, planar graphs satisfy the first two criteria, but fail to be robust since often adding a single random edge will make the graph non-planar. Indeed, most algorithms on planar graphs fail completely when run on a planar graph with a few perturbed edges and vertices. Next, one might try genus-bounded graphs, but they also suffer from similar problems since adding a single randomly connected vertex can arbitrarily increase the genus.

  A candidate graph family that fulfills all three conditions is the family of \textit{excluded minor} graphs, namely the graphs which do not have a fixed graph $H$ as a minor. This family encompasses several classes of naturally occuring networks. For example, trees which exclude $K_3$;planar graphs that capture the structure of two-dimensional maps exclude $K_5$ and $K_{3,3}$; and, series-parallel graphs that capture many network backbones exclude $K_4$~\cite{abraham2006object, flocchini2003routing}. Excluded minor graphs also have a history of deep results, including the series of Graph Minor papers by Robertson and Seymour.  

  In this paper, we provide efficient distributed algorithms for the class of excluded minor graphs which break the $\tO(\sqrt n+D)$ lower bound for general graphs, giving evidence that most practical networks admit efficient distributed algorithms. We show an $\tO(D^2)$ algorithm for MST and $(1+\e)$ approximate min-cut, among other results. For networks having low diameter, such as $D = \poly(\log n)$ or $D = n^{o(1)}$, our algorithms are optimal up to $\poly(\log n)$ or $n^{o(1)}$ factors, respectively. This is a significant improvement over previous MST and min-cut algorithms, which run in $\Omega(\sqrt n)$ time even on an excluded minor graph with $D = \poly(\log n)$, such as a planar graph with an added vertex attached to every other node.

  Our results use the framework of \textbf{low-congestion shortcuts}, introduced by~\cite{ghaffari2016distributed} and built on by~\cite{haeupler2016low}, which is a combinatorial abstraction to designing distributed algorithms. It introduces a simple, combinatorial problem involving \textbf{shortcuts} on a graph, and guarantees that a good \textbf{quality} solution to this combinatorial problem automatically translates to a simple, efficient distributed algorithm for MST and $(1+\e)$ approximate min-cut, among other problems; the concepts of shortcuts and quality will be defined later. Actually, the algorithm is the same regardless of the network or the graph family; the purpose of the combinatorial shortcuts problem is to prove that the algorithm runs \textit{efficiently} on the graph or family.
  

  To solve the shortcuts problem on excluded minor graphs, we appeal to the Graph Structure Theorem of Robertson and Seymour~\cite{robertson1986graph, robertson2003graph}. At a high level, the Graph Structure Theorem decomposes every excluded minor graph into a set of almost-planar graphs connected in a tree-like fashion. Our solution to the shortcuts problem is in fact a series of results, one for each step in the structure decomposition. We remark that our result is, to the best of our knowledge, the first in distributed computing to make use of the Graph Structure Theorem to claim a distributed algorithm is fast. The absence of such a preceding result in distributed computing is unsurprising, since algorithms working with the Graph Structure Theorem generally require computing the required decomposition beforehand, and no efficient distributed algorithm to do so is known. Even the best classical algorithm still takes $O(n^3)$ time~\cite{kawarabayashi2011simpler}, so even a sublinear distributed algorithm is still out of reach. However, our result is unique in that we merely show the \textit{existence} of a solution to the shortcuts problem in excluded minor graphs, and as a consequence, the simple algorithm of~\cite{haeupler2016low}---which does not look at any structure in the network graph, let alone compute a decomposition---is proven to run efficiently on excluded minor graphs.

  The fact that this algorithm does not actually compute the Graph Structure Theorem should be stressed further. One consequence is that the running time of this algorithm does not depend on any constants appearing in the Graph Structure Theorem. That is, while we can only \textit{prove} that the constants in the running time are bounded by (some functions of) the constants in the Graph Structure Theorem, the actual running time of the algorithm is likely to be much smaller. In fact, for most excluded minor networks, we expect the running time to be $\tO(D^2)$ with a small constant, or even $\tO(D)$. In contrast, algorithms that explicitly compute a Graph Structure Theorem decomposition have an inherent bottleneck in the form of the potentially huge constants of the Graph Structure Theorem.
  
  This paper is structured as follows. After the introduction, we begin with introducing the two main tools necessary for our main result, namely the Graph Structure Theorem and the low-congestion shortcuts framework. Then, we prove the existence of good shortcuts one step at a time, following the step-by-step construction in the Graph Structure Theorem.

\subsection{Outline of the Proof}

The goal of this section is to outline the proof of our main result, without delving into the technical details. Some concepts will be left undefined (e.g., shortcut quality, tree-restricted shortcuts), since the definitions are technical and require a lot of motivation beforehand. However, one can think of shortcuts as a combinatorial construction on a graph, and think of quality as a metric with which to measure a solution.

\begin{restatable}{theorem}{ShortcutsFramework}[Haeupler et al.~\cite{haeupler2016near,haeupler2016low}]
  \label{thm:ShortcutsFramework}
Suppose that a graph with diameter $D$ admits tree-restricted shortcuts of quality $q : \N\to\N$. Then, there is an $\tilde O(q(D))$-round distributed algorithm for MST and $(1+\e)$-approximate min-cut for that graph.
\end{restatable}

Our main technical result is showing the existence of good tree-restricted shortcuts in excluded minor graphs.

\begin{restatable}{theorem}{mainMinorShort}[Main Theorem]
  \label{thm:mainMinorShort} Every graph in a graph family excluding a fixed minor $H$ admits tree-restricted shortcuts of quality $q(d) = \tilde{O}(d^2)$. The constants in the big-$O$ depend only on the minor $H$.
\end{restatable}

Combining the above \Cref{thm:ShortcutsFramework} and \Cref{thm:mainMinorShort}, we get out main result.
\BC
  \label{thm:mainMinorAlgo}
There exists an $\tilde O(D^2)$-round distributed algorithm for MST and $(1+\e)$-approximate min-cut, for any $\e>0$, on graph networks excluding a fixed minor.
\EC

For excluded minor graphs of diameter $n^{o(1)}$, as is the case for most practical networks, our algorithms also run in $n^{o(1)}$ time, which is optimal up to lower order terms. This is a significant improvement over the previously known $\tilde\Omega(\sqrt n)$ time algorithms and it avoids the $\tilde{\Omega}(\sqrt{n})$ lower bound for general graphs, even when they have $n^{o(1)}$ diameter.

\BC
There exists an $n^{o(1)}$-round distributed algorithm for MST and $(1+\e)$-approximate min-cut, for any $\e>0$, on graph networks with diameter $n^{o(1)}$ and excluding a fixed minor.
\EC

In order to work with excluded minor families, we appeal to the Robertson-Seymour Graph Structure Theorem. At a high level, this theorem states that every graph in an excluded minor family can be decomposed into a set of graph \textbf{almost embeddable} in a bounded genus surface that are glued together in a tree-like fashion. Naturally, our approach is to first construct good-quality shortcuts for the entire family of almost embeddable graphs, and then modify them in a robust manner as they are patched together in the composition. While this approach works in general, the patching required is very involved because of various interactions between the many ingredients involved in the decomposition.
For example, one step in the construction of an almost embeddable graph is the addition of an ``apex'' vertex that connects arbitrarily to all previous vertices. While the addition of only one vertex appears harmless at first glance, observe that the diameter can shrink arbitrarily, e.g., to $2$ if the apex is connected to all other vertices. If the graph without the apex has large diameter $D$, and its shortcuts solution leads to an $\tO(D^2)$-round algorithm, this same algorithm will not suffice on the graph with the apex, which can have diameter $2$. A lot of technical effort goes into reconstructing shortcuts upon the addition of an apex, in order to handle the arbitrary decrease in graph diameter.
Hence, as a consequence of all these difficulties, we settle for $\tilde O(d^2)$-quality shortcuts, and leave the improvement to $\tilde O(d)$-quality shortcuts as an open problem.

\subsection{Related Work}

Work on global network problems in the distributed setting was started by Gallager, Humblet and Spira~\cite{gallager1983distributed} who gave a $O(n \log n)$-round algorithm to compute the MST. The algorithm was subsequently improved by Awerbuch~\cite{awerbuch1987optimal} to an ``optimal'' $O(n)$ rounds. However, Peleg and Awerbuch~\cite{awerbuch1989randomized} noted that a more useful notion of round complexity would parametrize on both the number of nodes $n$ and the diameter $D$ since $D \ll n$ for many practical networks. This influenced a substantial amount of work that followed, culminating in an $\tilde{O}(D + \sqrt{n})$ distributed algorithm for many optimizations problems in the CONGEST setting. Examples of such problems include the MST~\cite{garay1998sublinear, khan2008fast,elkin2017simple}, $(1+\e)$-approximate Maximum Flow~\cite{ghaffari2015near}, Minimum Cut~\cite{nanongkai2014almost, ghaffari2013distributed}, Shortest Paths, and Diameter~\cite{nanongkai2014distributed,frischknecht2012networks,holzer2012optimal,Lenzen:2013,Lenzen:2015,izumi2014time, henzinger16almost, elkin2017distributed, huang2017distributed}. The mentioned $\tilde{O}(D + \sqrt{n})$ round complexity is \textbf{existentially optimal} for all of these problems, as there exists a family of graphs for which any algorithm must take a matching $\tilde \Omega(D+\sqrt n)$ rounds~\cite{elkin2006unconditional, sarma2012distributed}.

Some early work that tried to circumvent the $\Omega(\sqrt{n})$ bound was done by Khan and Pandurangan\cite{khan2008fast} who argued that their MST algorithm performed in an \textbf{universally optimal} manner on some restricted classes of graphs. In particular, they gave a $\tilde{O}(D)$ round MST algorithm for (1) unit disk networks where weights match the distances, and (2) networks with IID random weights. Their approach used a novel parameter called ``local shortest path diameter'' that takes the edge weights into account. This makes their approach unsuitable for arbitrary weights and topologically constrained networks such as planar graphs.

However, an alternative approach has recently made progress for global optimization problems on restricted graph classes. Ghaffari and Haeupler showcased a distributed MST and $(1+\e)$-approximate min-cut algorithm that runs in $\tilde{O}(D)$ rounds on planar graphs~\cite{ghaffari2016distributed}, which was later simplified and generalized to bounded treewidth graphs and bounded genus graphs~\cite{haeupler2016near}. All of these results use the aforementioned low-congestion shortcuts framework.

The Graph Structure Theorem originates from a series of deep results on graph minor theory by Robertson and Seymour~\cite{robertson1986graph, robertson2003graph}. It provides a structural decomposition to all excluded minor graphs, transforming a negative property---not containing a minor---to a positive and constructive property that is more useful for algorithm design. The original statement of the theorem only states that such a structure exists, but in a later breakthrough, Demaine et al.\ developed a polynomial-time algorithm to compute the decomposition guaranteed by the theorem\cite{demaine2005algorithmic}. Since then, the graph structure decomposition has found numerous algorithmic applications on excluded minor graphs, such as polynomial-time approximation schemes~\cite{grohe2003local, demaine2005algorithmic},  subexponential algorithms~\cite{demaine2005subexponential}, graph coloring~\cite{devos2004excluding}, and computing separators~\cite{grohe2003local, abraham2006object}. Since then, simpler proofs of the Graph Structure Theorem have been discovered~\cite{kawarabayashi2011simpler}, as well as more efficient algorithms, with the fastest known one running in time $f(H)\cd n^3$ for a function $f$ depending only on the excluded minor $H$~\cite{kawarabayashi2011simpler}. However, we believe this is the first time the theorem is used for an algorithmic result in the distributed setting.

\subsection{Preliminaries}

For a graph $G$, let $V(G)$ and $E(G)$ denote the vertices and edges, respectively. Given $P \subseteq V(G)$, $G[P]$ denotes the induced subgraph, namely, the one obtained by removing $V(G) \setminus P$ from $G$. Finally, when $G$ is the underlying network graph, we always assume that $G$ is connected and contains no self-loops (which can be ignored in the distributed setting anyway).

\fullOnly{
\subsubsection{CONGEST model}

While we only use the classical CONGEST model indirectly, via \Cref{thm:ShortcutsFramework} of \cite{haeupler2016low}, we will state its assumptions for context. A network is given by a connected graph $G=(V,E)$ with $n$ nodes and diameter $D$. Communication proceeds in synchronous rounds. In each round, each node can send a different $O(\log n)$ bit message to each of its neighbors. Local computations are free and require no time. Nodes have no initial knowledge of the topology $G$, except that we assume that they know $n$ and $D$ up to constants. One can easily remove these assumptions by distributively computing these parameters in $O(D)$ time, which is negligible in our context.
}

\subsubsection{Graph Structure Theorem}

\begin{figure}[H]
  \centering
  \shortOnly{\includegraphics[width=0.76\textwidth]{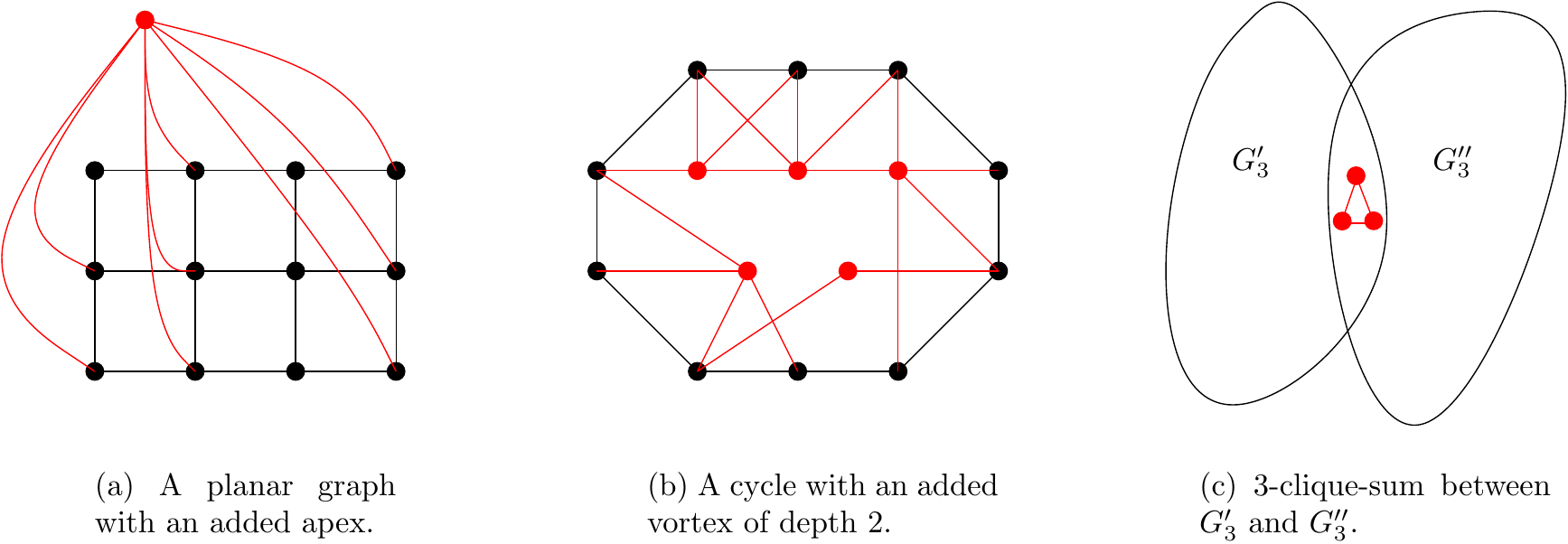}}
  \fullOnly{\includegraphics[width=0.90\textwidth]{figures/graph-structure.pdf}}
  \caption{Ingredients of the Graph Structure Theorem}
  \label{figure:graph-structure-theorem}
\end{figure}

In this section we introduce Robertson and Seymour's  Graph Structure Theorem, following the survey of Lov\'{a}sz~\cite{lovasz2006graph}. This theorem is instrumental in our shortcut construction, since it provides structure for all graphs that are $H$-free, for any minor $H$. At a high level, the theorem says that every $H$-free graph can be glued together in a tree-like fashion from graphs that can be ``almost'' embedded in a fixed surface. To elaborate on this statement, we need a few definitions. The first definition, $k$-clique-sum, captures the tree-like structure of the graph.

\BD[$k$-clique-sum] Let $G_1$ and $G_2$ be two graphs, and let $S_i\s G_i$ be a $k$-clique for $i=1,2$. Let $G$ be obtained by identifying  $S_1$ with $S_2$ and deleting some (possibly none, possibly all) edges between the nodes in $S_1=S_2$. We say that $G$ is a \textbf{$k$-clique-sum} of $G_1$ and $G_2$. More generally, $G$ is a $k$-clique-sum of $G_1,G_2,\lds,G_\ell$ if $G$ is formed by starting with $G_1$ and iteratively taking the $k$-clique-sum of the current graph with $G_i$, for $i=2,\lds,\ell$ in that order.
\ED

The next few definitions classify the graphs which are almost embeddable on a surface. We start with the three main ingredients in constructing such a graph, and then define what it means to be almost-embeddable.

\BD[Apex] Define \textbf{adding an apex to graph $G$} as follows: create a new vertex called the \textbf{apex}, and connect it to an arbitrary subset of the vertices in $G$.
\ED

\BD[Surface of genus $g$]
A graph $G$ has genus $g$ if there is a \textbf{2-cell embedding} in a surface of genus $g$. In other words, this means: (i) there exists an oriented or unoriented surface (i.e., 2-manifold) $\Sigma$ of genus $g$, (ii) vertices of $G$ are mapped to distinct points of $\Sigma$, (iii) edges are mapped to simple paths whose interiors do not contain any vertices of $G$ nor do path interiors mutually intersect, and (iv) each face defined by such embedding is homeomorphic to a unit disk, i.e., contains no holes or handles in it.
\ED

\BD[Vortex~\cite{lovasz2006graph}]\label{defn:Vortex}
Let $G$ be a graph with a 2-cell embedding. Let $C$ be a cycle in $G$ that corresponds to a face on the surface. Call a continuous interval on the cycle an \textbf{arc}. Select a family of arcs on $C$ so that each node is contained in at most $k$ of these arcs. For each arc $A$, create a new node $v_A$ and connect $v_A$ to a subset of the vertices in $C$ that lie on arc $A$. Such nodes $v_A$ are called \textbf{internal vortex nodes}. Finally, for any two arcs $A$ and $B$ sharing a common vertex in $C$, we may add the edge $\{v_A,v_B\}$. We call this operation \textbf{adding a vortex of depth $k$ to cycle $C$}.
\ED

\BD A graph $G$ is \textbf{$(q,g,k,\ell)$-almost-embeddable} if it can be constructed according to the three steps below.
\BE
 \setlength{\itemsep}{0pt}
 \im[(i)] Start with a graph $G'$ embedded on a surface of genus at most $g$.
 \im[(ii)] We select at most $\ell$ faces of $G'$ and add a vortex of depth at most $k$ to each of them. Call the result $G''$.
 \im[(iii)] We add $q$ apices to $G''$, connected arbitrarily to vertices in $G''$ and to each other, and obtain the desired graph $G$.
\EE
For simpler notation, we say a graph is \textbf{$h$-almost-embeddable} if it is $(h,h,h,h)$-almost embeddable.
\ED

By this definition, the planar graphs are precisely the $(0,0,0,0)$-almost-embeddable graphs, and the genus-$g$ graphs are precisely the $(0,g,0,0)$-almost-embeddable graphs. Later on, we will study the planar graphs with added vortices, in particular the $(0,0,k,\ell)$-almost-embeddable graphs for constants $k$ and $\ell$.

As a final ingredient to the Graph Structure Theorem, we construct a graph family $\mathcal{L}_k$ as follows. 
\BD
Let $\m L_k$ denote all graphs that can be represented as a $k$-clique-sum of $k$-almost-embeddable graphs. That is, a graph $G$ is in $\m L_k$ if there exist $k$-almost-embeddable graphs $G_1,G_2,\ldots,G_\ell$ such that $G$ is a $k$-clique-sum of $G_1,G_2,\ldots,G_\ell$.
\ED

In other words, take any set of $k$-almost-embeddable graphs $G_1,G_2,\lds,G_\ell$ for $\ell\ge1$, and let $G$ be their $k$-clique-sum. 
Construct a graph $G$ by repeatedly taking a $k$-clique-sum operation between multiple $G_i$'s constructed using step (i)--(iii) and connect them in a tree-like fashion.
Define $\m L_k$ as precisely all graphs $G$ that can be constructed in this way.

Finally, we present the Graph Structure Theorem, which states that for any $H$, there is a $k$ such that $\m L_m$ includes (but does not exactly characterize) all graphs that are $H$-free.

\begin{theorem}[Graph Structure Theorem]\label{thm:GraphStructure}
  For every graph $H$ there is a fixed integer $k = k(H)$ such that any $H$-free graph $G$ is contained in $\mathcal{L}_k$.
\end{theorem}

Below, we include additional terminology on clique-sums and vortices used in the shortcut construction.

\BD[Vortex terminology] Let $C$ be a cycle of $G$, and add a vortex of depth $k$ to $C$, following \Cref{defn:Vortex}. Let $v_1,v_2,\lds$ be the vertices created when adding a vortex of depth $k$ to $C$. The vertices in $C$ form the \textbf{vortex boundary}, and the added vertices $v_1,v_2,\lds$ are called \textbf{inside the vortex} and \textbf{internal vortex} nodes. Moreover, suppose an internal vertex $v_i$ corresponds to \textbf{arc} $A_i$ of $C$ in the vortex construction. Define the \textbf{vortex decomposition} to be the map $\m P$ satisfying $\m P(v_i)=A_i$. Finally, if $G$ is embedded on a closed surface such that $C$ forms a face in the embedding, then that face is called the \textbf{vortex face}.
\ED

\BD[$k$-Clique-sum decomposition tree] \label{defn:CliqueSumDecomposition}
Let the graph $G$ be constructed as the $k$-clique-sum of subgraphs $B_1, B_2, \ldots, B_\ell$. The subgraphs $B_i \subseteq G$ are denoted as \textbf{bags}. A \textbf{$k$-clique-sum decomposition tree} of $G$ is a tree $\dectree$ whose vertices $V(\dectree)$ are identified with bags $B_i$. The edges of the decomposition $f\in E(\dectree)$ correspond to a clique in two of the bags, with possibly some removed edges. Therefore, we refer to them as \textbf{partial $k$-cliques} $C_f$. The decomposition satisfies the following properties:
\BE
    \setlength{\itemsep}{0pt}
    \im $\Bigcup_{i\in\dectree}V(B_i)=V(G)$.
    \im For all $i \in V(\dectree)$, $B_i \subseteq G$.
    \im For all $f=\{i,j\}\in E(\dectree)$, $B_i\cap B_j=C_f$.
    \im For all $v\in V(G)$, the set $\{i\in V(\dectree)\mid v\in V(B_i)\}$ is connected in $\dectree$.
    \im For all $e\in E(G)$, there exists $i\in V(\dectree)$ with $e\in E(B_i)$.
 \EE
\ED

We conclude with a statement that the above clique-sum decomposition tree captures all possible ways to take clique-sums of graphs.
\begin{fact}\label{fact:CliqueSumDecomposition}
Let a graph $G$ be the $k$-clique-sum of graphs from a family $\m F$. Then, $G$ has a $k$-clique-sum decomposition tree whose bags are graphs in $\m F$.
\end{fact}

\subsubsection{Tree-restricted Shortcuts}\label{section:TreeRestrictedShortcuts}

This section introduces the shortcut framework from \cite{ghaffari2016distributed} and \cite{haeupler2016low}. We state the definitions and concisely repeat the motivations behind them.

Imagine solving the following problem in a distributed manner: \textit{Each node in the network is assigned a number $x_v$. The network is partitioned into a number of disjoint individually-connected parts and each node wants to compute the minimum/maximum/sum of $x_v$ between all nodes in its own part}. The above subproblem often occurs in algorithm design. Notably, it appears in Boruvka's Minimum Spanning Tree algorithm~\cite{nevsetvril2001otakar}. A naive solution would spread the information about $x_v$ independently inside each part, because the problem statement basically makes them independent. However, this will incur performance penalties if the diameter of the parts in isolation is much larger than the diameter of the entire graph. For instance, a wheel graph has $\Theta(1)$ diameter, while a part containing all the outer nodes has $\Theta(n)$ diameter. This leads us to the idea of ``helping'' a part by assigning it additional edges that it can use to spread information. We call these edges ``shortcuts''.

\begin{definition}[General Shortcuts]
  \label{def:generalshortcut}
  Let $G = (V, E_G)$ be an undirected graph with vertices subdivided into \textbf{pairwise disjoint and connected} subsets $\mathcal{P} = (P_1, P_2, ..., P_{|\mathcal{P}|}), P_i \subseteq V$. In other words, $G[P_i]$ is connected and $P_i \cap P_j = \emptyset$ for $i \neq j$.
  The subsets $P_i$ are called \textbf{parts}.
  We define a \textbf{shortcut} $\mathcal{H}$ as a tuple of \textbf{shortcut edges} $(H_1, H_2, \ldots, H_{|\mathcal{P}|})$, $H_i \subseteq E(G)$.
\end{definition}

The shortcut framework has particularly nice properties if the shortcut edges are restricted to some tree $T$, typically the BFS tree. In particular, they can be near-optimally constructed in a distributed and uniform manner on any network. We first define this structured version of shortcuts, enumerate the quality measures for them, and finally state the relevant theorems.

\begin{definition}[Tree-Restricted Shortcuts]
  Let $\mathcal{H} = (H_1, H_2, ..., H_{|\mathcal{P}|})$ be a shortcut on the graph $G = (V, E_G)$ with respect to the parts $\mathcal{P} = (P_i)_i$. Given a  spanning tree $T = (V, E_T) \subseteq G$ we say that a shortcut $\mathcal{H}$ is \textbf{$T$-restricted} if for each part $P_i \in \mathcal{P}$, its shortcut edges $H_i \subseteq E_T$. In other words, every edge of $H_i$ lies on the tree $T$.
\end{definition}

Some shortcuts are preferable to others. We define \textbf{congestion, block}, and \textbf{quality} parameters to measure their usefulness. Firstly, the congestion parameter measures how many shortcuts use a particular edge, thereby resulting in over-congestion.
\begin{definition}[Congestion]
  Let $\mathcal{H} = (H_1, H_2, ..., H_{|\mathcal{P}|})$ be a shortcut on $G = (V, E_G)$. Define a congestion over an edge $e$ as $c_e := |\{ i : e \in H_i \}|$, the number of shortcuts using an edge. Finally, define the \textbf{congestion of the shortcut} to be $\max_{e \in E_G} c_e$, the maximum congestion of any edge.
\end{definition}

Another beneficial property we would like is that the subgraphs $G[P_i] + H_i$ have small diameter. By $G[P_i] + H_i$ we mean the subgraph constructed by taking the induced subgraph $G[P_i]$ and then adding all the edges in $H_i$ as well as any of their endpoints not contained in $P_i$. We will measure diameter of $G[P_i] + H_i$ indirectly, by defining a \textbf{block parameter} $b$ that essentially measures the number of different subtrees (components) that exist in $H_i$. We use it in the following way: note that (1) $G[P_i]$ is connected by definition, (2) $G[P_i] + H_i$ effectively looks like $b_i$ interconnected subtrees, (3) each subtree has diameter $O(D)$ since $T$ is typically a spanning tree whose height is at most $D$, the diameter of $G$. From these properties we can easily conclude that $G[P_i] + H_i$ has a controlled, $O(b_i D)$, diameter.

\begin{definition}[Block parameter]
  For a shortcut $\mathcal{H} = (H_1, H_2, ..., H_{|\mathcal{P}|})$, fix a part $P_i$ and consider the connected components of the spanning subgraph $(V, H_i)$. If a connected component contains a node in $P_i$ we call it a \textbf{block component}. They correspond to subtrees of $T$. We define that $\mathcal{H}$ has \textbf{block parameter} $b$ if no part has more than $b$ block components.

  In general, for functions $b,c: \N\to\N$, a graph $G$ is said to \textbf{admit} tree-restricted shortcuts with block parameter $b$ and congestion $c$ if for \textbf{any} spanning tree $T$ with diameter at most $d_T$ and any family of parts there exists a $T$-restricted shortcut with  block parameter $b(d_T)$ and congestion $c(d_T)$. For ``good'' tree-restricted shortcuts, we typically expect a block parameter of $\tO(1)$ and a congestion of $\tO(d_T)$.
 
  Similarly, a family of graphs admits tree-restricted shortcuts with block parameter $b$  and  congestion $c$ if all graphs in the family individually admit them.
\end{definition}


We now define the \textbf{quality} of a tree-restricted shortcut. The terminology of admitting tree-restricted shortcuts carries over to quality.

\BD[Quality]
The \textbf{quality} of a tree-restricted shortcut is the function $q:\N\to\N$ defined by $q(d) = b(d) \cd d + c(d)$.
\ED

We now restate the central theorem from the low-congestion shortcuts framework.

\ShortcutsFramework*

Note that $D$ replaces $d_T$ in the argument to the functions $b$ and $c$. This is because the theorem takes $T$ to be a BFS tree of the network graph, which has diameter at most $D$.

Finally, it is known from~\cite{ghaffari2016distributed,haeupler2016near} that good tree-restricted shortcuts exist in planar graphs and bounded treewidth graphs. We will be using this fact in a later section.

\begin{theorem}[\hspace{1sp}\cite{ghaffari2016distributed}]
  \label{theorem:shortcut-existence-on-bounded-genus}
  The family of planar graphs admits tree-restricted shortcuts with block parameter $O(\log d_T)$ and congestion $O(d_T\log d_T)$.
\end{theorem}

\begin{theorem}[\hspace{1sp}\cite{haeupler2016near}]
  \label{theorem:shortcut-existence-on-bounded-treewidth}
  The family of graphs of treewidth at most $k$ admits tree-restricted shortcuts with block parameter $O(k)$ and congestion $O(k\log n)$.
\end{theorem}

\section{Shortcuts in Excluded Minor Graphs}

Our main result extends~\Cref{theorem:shortcut-existence-on-bounded-genus} and~\Cref{theorem:shortcut-existence-on-bounded-treewidth} to excluded minor graphs, showing that any family of graphs excluding a fixed minor has good tree-restricted shortcuts. We repeat \Cref{thm:mainMinorShort} with a bit of extra detail.

\begin{restatable}{theorem}{mainMinorLong}[Main Theorem, Extended Version]
  \label{thm:mainMinorLong} The family of graphs excluding a fixed minor $H$ admits tree-restricted shortcuts of quality $q(d) = \tO(d^2)$. More generally, the family admits block parameter $b(d) = O(d)$ and congestion $c(d) = O(d \log n + \log^2 n)$. The constants in the big-$O$ depend only on the minor $H$. 
\end{restatable}

Using the shortcuts framework of~\Cref{thm:ShortcutsFramework}, the above theorem translates to the algorithmic result of~\Cref{thm:mainMinorAlgo}.

\subsection{Two Parts of the Proof}

Recall that the Graph Structure Theorem says that any excluded minor graph can be represented as a $k$-clique-sum of $k$-almost-embeddable graphs, for some constant $k$ depending on the excluded minor. As with most results utilizing the Graph Structure Theorem, our proof is split into two parts, one handling the $k$-clique-sums and one for the $k$-almost-embeddable graphs.

Our proof has two main components, namely, \Cref{thm:CliqueSum} and \Cref{thm:almostEmbeddable} that we state below. It should be clear that they are sufficient to prove the main technical result, \Cref{thm:mainMinorLong}.

\hdr{Clique Sums Part}
In the $k$-clique-sums part, we show that if a family of graphs admits shortcuts with good quality, then so does any $k$-clique-sum of graphs from this family, for any constant $k$. In other words, having good tree-restricted shortcuts is a property robust under taking $k$-clique-sums for a fixed integer $k$. The theorem below is proved in \Cref{section:CliqueSum}.

\begin{restatable}{theorem}{CliqueSum}[Shortcuts in Clique Sums] \label{thm:CliqueSum}
Let $\m F$ be a family of graphs that admits tree-restricted shortcuts with block parameter $b_{\m F}$ and congestion $c_{\m F}$. Let $G$ be a $k$-clique-sum of graphs in $\m F$. Then $G$ admits tree-restricted shortcuts with block parameter $b_G(d) \le 2k + O(b_{\m F}(d_T))$ and congestion $c_G(d) \le O(k \log^2 n) + c_{\m F}(d_T)$.
\end{restatable}

To prove the full result, we use \Cref{thm:CliqueSum} with $\m F$ as the family of $k$-almost-embeddable graphs, which we show admits tree-restricted shortcuts with block parameter and congestion $\tO(d)$. Plugging in these parameters, we obtain $b_G(d)= 2k + \tilde{O}(d)$ and $c_G(d) = O(k \log^2 n) + \tilde{O}(d)$ for the final result, which are both $\tilde{O}(d)$ since $k$ is a constant. Note that \Cref{thm:CliqueSum} does not assume that $\m F$ is any particular family, so it may be of independent interest.

\hdr{Almost Embeddable Part} The second part of the proof establishes good quality shortcuts for $k$-almost-embeddable graphs, namely the theorem below, proved in \Cref{section:AlmostEmbeddable}.

\shortOnly{\vspace{-2mm}}

\begin{restatable}{theorem}{BuildingBlocks}[Shortcuts in Almost Embeddable Graphs] \label{thm:almostEmbeddable}
An $(q,g,k,\ell)$-almost-embeddable graph $G$ admits tree-restricted shortcuts with block parameter $b(d)=O(q + (g+1)k\ell^2d)$ and congestion $c(d)=O(q + k\ell^2d(g+\log n))$.
\end{restatable}

\shortOnly{\vspace{-3mm}}
\hdr{Putting Them Together}
The main theorem, restated below, follows immediately from \Cref{thm:CliqueSum} and \Cref{thm:almostEmbeddable}.

\shortOnly{\vspace{-2mm}}
\mainMinorLong*
\fullOnly{
\BP
By \Cref{thm:GraphStructure}, there is a constant $k$ such that the family of $H$-free graphs is contained in $\m L_k$, so it suffices to prove the claim for $\m L_k$.
Let $\m F$ be the family of $k$-almost-embeddable graphs. By \Cref{thm:almostEmbeddable}, $\m F$ admits tree-restricted shortcuts with block parameter $b_{\m F}(d)=O(d)$ and congestion $c_{\m F}(d)=O(d \logn)$. Plugging in $\m F$, $b_{\m F}$, and $c_{\m F}$ into \Cref{thm:CliqueSum}, we conclude that $\m L_k$ admits tree-restricted shortcuts with block parameter $O(d)$ and congestion $O(d\log n + \log^2 n)$, as desired.
\EP
}

\subsection{Shortcuts in Clique Sum Graphs}\label{section:CliqueSum}

In this section, we prove \Cref{thm:CliqueSum}
, restated below. \shortOnly{Due to space constraints, it appears in the full version of the paper, attached at the end of this extended abstract.} The proof is completely self-contained and its inner working will not be reused in following sections. 

\CliqueSum*

\fullOnly{
\hdr{Local and Global Shortcuts}
The intuition behind our construction is as follows. Let $G$ be a $k$-clique-sum of graphs in $\m F$, and consider a $k$-clique-sum decomposition tree $\dectree$ of $G$. Its existence is guaranteed by Fact \ref{fact:CliqueSumDecomposition}. Consider a part $P\s V(G)$, which could either span much of a single bag in $\dectree$, or traverse through multiple bags, or both. As a result, we construct two types of shortcuts---\textbf{local} shortcuts and \textbf{global} shortcuts---to handle each case separately. At a high level, local shortcuts, which are constrained within a single bag, are meant to deal with parts that behave wildly within a bag, while global shortcuts, which can span multiple bags, treat parts that stretch across many different bags. In particular, for each part $P$, we specify one bag on which we construct local shortcuts for $P$, and let global shortcuts handle the rest. The shortcut for $P$ is simply the union of the local and global shortcuts.

Root $\dectree$ at an arbitrary bag, and define $d_{\dectree}$ to be the depth of the rooted tree $\dectree$.
We first prove a weaker result whose global shortcut depends on the value of $d_{\dectree}$ in its congestion, then later show how to ``compress'' $\dectree$ to a low depth independent of $d_{\dectree}$, thereby removing the dependence of $d_{\dectree}$.

\begin{figure}
  \centering
  \includegraphics[width=0.55\textwidth]{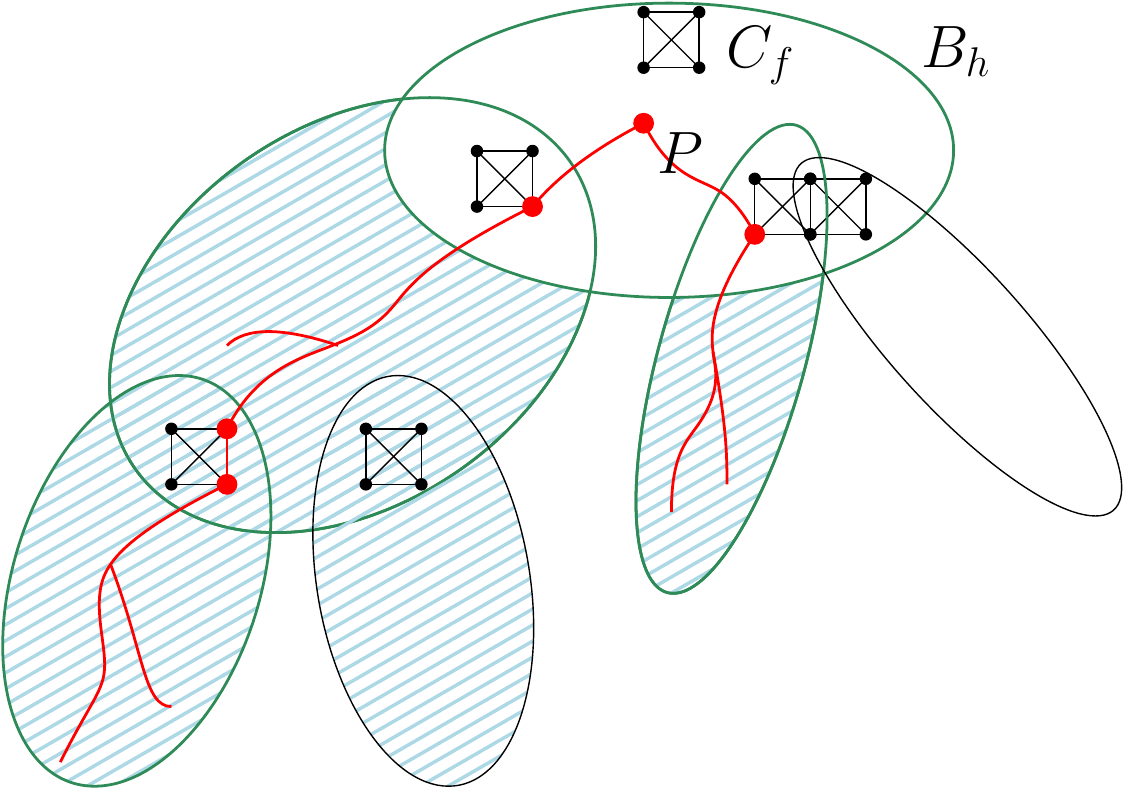}
  \caption{Global shortcut construction. The part $P$ is shown in red. The global $T$-restricted shortcut is the intersection of $T$ (not shown) with the shaded region. $C_f$ denotes the partial clique leading to the parent of $h$, which is not used in the global shortcut.}
  \label{figure:clique-sum-shortcuts}
\end{figure}


\begin{lemma} \label{lem:CliqueSumBad}
  Let $\m F$ be a family of graphs that admits tree-restricted shortcuts with block parameter $b_{\m F}$ and congestion $c_{\m F}$. Let $G$ be a $k$-clique-sum of graphs in $\m F$ with decomposition tree $\dectree$. Then $G$ admits tree-restricted shortcuts with block parameter $b_G(d_T) \le k + b_{\m F}(d_T)$ and congestion $c_G(d_T) \le k \textcolor{blue}{d_{\dectree}} + c_{\m F}(d_T)$. (Note the dependence on \textcolor{blue}{$d_{\dectree}$}, the depth of the decomposition tree $\dectree$, which is unrelated to $d_T$, the diameter of the spanning tree $T$.)
\end{lemma}
\begin{proof}
Let $T$ be an arbitrarily rooted spanning tree of $G$ of diameter $d_T$.
Take a $k$-clique-sum decomposition tree $\dectree$, root it at an arbitrary bag, and suppose that the rooted tree has depth $d_{\dectree}$. In the rooted setting, define the set $\desc(i)\s V(\dectree)$ for $i\in V(\dectree)$ to be $i$ along with all of its descendants in $\dectree$.

Consider a part $P\s V(G)$. Since $P$ is connected, we know,  by properties (4) and (5) of \Cref{defn:CliqueSumDecomposition}, that the set of bags $S_P:=\{ j\in V(\dectree) \mid V(B_j)\cap P\ne\emptyset\}$ is connected  in $\dectree$. Therefore, the lowest common ancestor, denoted by $h_P$, of $S_P$ is also inside $S_P$. Similarly, for an edge $e \in E(G)$ we can define the set of bags that contain that edge $S_e := \{ j \in V(\dectree) \mid e \in E(B_j) \}$ and its lowest common ancestor $h_e \in S_e$.

\hdr{Global Shortcuts} See~\Cref{figure:clique-sum-shortcuts}.
The construction of the global shortcut is simple. For each edge $f'$ to a child $i$ of $h_P$ such that $P\cap V(C_{f'})\ne\emptyset$, allow part $P$ to use all edges in $\left( \Bigcup_{j\in\desc(i)}E(B_j) \cap T \right ) \setminus E(B_{h_P})$. Informally, the global shortcut ``takes care'' of all vertices in $P$ except for those in $B_{h_P}$, which leaves constructing the local shortcut for $P$ in $B_{h_P}$. More precisely, consider the roots of the block components of $P$ when using only the global shortcut: they are restricted to $B_{h_P}$.

We now argue about the congestion. Consider an edge $e \in E(G)$, and let $\mathcal{B}$ be the set of bags on the $\dectree$-root-path to $h_e$, including $h_e$. Clearly, $|\mathcal{B}| \le d_{\dectree}$. Edge $e$ can only be assigned to parts that contain a vertex in the partial-clique on a parent edge of a bag in $\mathcal{B}$. Hence its congestion is at most $k |\mathcal{B}| \le k d_{\dectree}$.

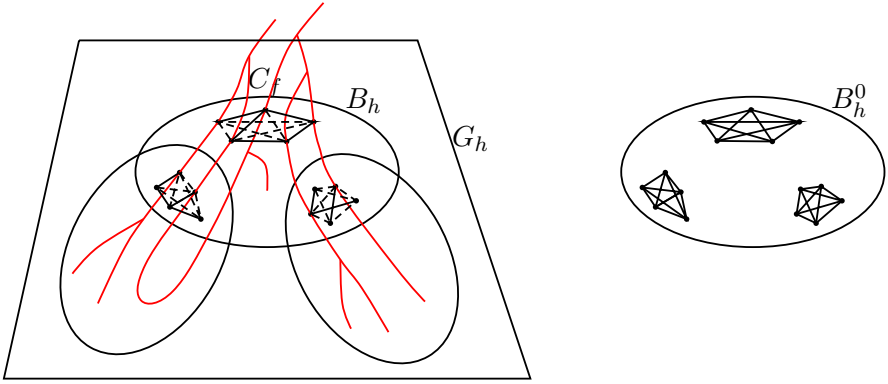
\begin{figure}
\centering
\begin{tikzpicture}[scale=.5]

\draw[rotate=60, line width=.7pt]  (0.5,19) ellipse (3 and 2);
\draw[ line width=.7pt]  (-13,12) ellipse (3.5 and 2);
\draw[rotate=120, line width=.7pt]  (13.5,4) ellipse (3 and 2);

\node[fill=black, circle, scale=.2] (v22) at (-11.7313,11.5405) {};
\node[fill=black, circle, scale=.2] (v10) at  (-11.1805,11.6144) {};
\node[fill=black, circle, scale=.2] (v23) at (-10.6297,11.2319) {};
\node [fill=black, circle, scale=.2](v11) at (-11.8384,10.88) {};
\node[fill=black, circle, scale=.2] (v24) at (-11.3182,10.6352) {};
\node[fill=black, circle, scale=.2, label=above:$C_f$] (v5) at (-13.0471,13.6519) {};
\node[fill=black, circle, scale=.2] (v1) at (-14.3017,13.3306) {};
\node[fill=black, circle, scale=.2] (v3) at (-13.9498,12.8256) {};
\node[fill=black, circle, scale=.2] (v9) at (-12.481,12.8104) {};
\node[fill=black, circle, scale=.2] (v8) at (-11.7619,13.3306) {};
\node[fill=black, circle, scale=.2] (v2) at (-15.3268,11.9842) {};
\node[fill=black, circle, scale=.2] (v4) at (-14.9137,11.464) {};
\node[fill=black, circle, scale=.2] (v21) at (-14.7607,10.7449) {};
\node[fill=black, circle, scale=.2] (v19) at (-15.9388,11.5864) {};
\node[fill=black, circle, scale=.2] (v20) at (-15.5869,11.0662) {};
\node (v12) at (-16.8378,9.891) {};
\node (v13) at (-16.2317,9.5461) {};
\node (v14) at (-10.4738,8.8982) {};
\node (v15) at (-9.5855,9.4625) {};
\node (v16) at (-12.429,12.3502) {};
\node (v17) at (-13.519,12.5279) {};
\draw [red,line width=.7pt] plot[smooth, tension=.7] coordinates {(v1) (v2) (-16.2874,10.745)
 (v12) (-17.5,8.5)};
\draw [red,line width=.7pt] plot[smooth, tension=.7] coordinates {(-16.2874,10.745) (-17.5,10) (-18.1768,9.2921)};
\draw [red,line width=.7pt] plot[smooth, tension=.7] coordinates {(v3) (v4) (v13) (-16.3432,8.5715) (-15.5022,8.8302) (-14.47,10.4848) (v17) (v5)};

\draw [red,line width=.7pt] plot[smooth, tension=.7] coordinates {(v1) (-13.8231,14.1606) (-13.4755,15.09) (-12.7717,16.0577)};
\draw [red,line width=.7pt] plot[smooth, tension=.7] coordinates {(v3) (-13.552,13.8432) (-13.4755,15.09)
};
\node (v6) at (-12.2056,15.6486) {};
\node (v7) at (-11.9302,14.6694) {};
\draw [red,line width=.7pt] plot[smooth, tension=.7] coordinates {(v5) (-12.6187,14.8989) (v6) (-11.5,16.5)};
\draw [red,line width=.7pt] plot[smooth, tension=.7] coordinates {(v6) (v7) (v8)};
\draw [red,line width=.7pt] plot[smooth, tension=.7] coordinates {(v7) (-12.4847,13.5717) (v9)};
\draw [red,line width=.7pt] plot[smooth, tension=.7] coordinates {(v8) (-11.6242,12.6943) (v10) (v15) (-8.8017,8.5503)};
\draw [red,line width=.7pt] plot[smooth, tension=.7] coordinates {(v9) (-12.2668,11.761) (v11) (-11.052,9.6643)  (v14) (-9.7729,7.7372)};
\draw [red,line width=.7pt] plot[smooth, tension=.7] coordinates { (-11.052,9.6643)  (-11,8.5) (-10.7674,7.829)};

\node at (-13,12) {};
\node (v18) at (-13.0549,12.2004) {};
\draw [red,line width=.7pt] plot[smooth, tension=.7] coordinates {(v17) (v18) (-13,11.5)};

\draw [black,line width=.7pt] (v2.center) -- (v19.center) -- (v20.center) -- (v21.center);
\draw [black,line width=.7pt](v1.center) -- (v5.center) -- (v3.center) -- (v9.center) -- (v8.center) -- (v5.center);
\draw [black,densely dashed,line width=.7pt](v1.center) -- (v3.center) -- (v8.center) -- (v1.center) -- (v9.center) -- (v5.center);
\draw [black,densely dashed,line width=.7pt](v19.center) -- (v4.center) -- (v20.center) -- (v2.center) -- (v21.center) -- (v19.center) (v2.center) -- (v4.center) -- (v21.center);
\draw [black,line width=.7pt] (v20.center) edge (v4.center);
\draw [black,line width=.7pt ](v22.center) -- (v11.center) -- (v23.center) (v10.center) -- (v24.center);

\draw [black,densely dashed,line width=.7pt](v22.center) -- (v10.center) -- (v23.center) -- (v24.center) -- (v11.center) -- (v10.center) (v22.center) -- (v24.center) (v22.center) -- (v23.center);
\draw [black,line width=.7pt] (-18,15.5) node (v25.center){} -- (-9,15.5) -- (-6,6.5) -- (-20,6.5) --  (-18,15.5);
\node[scale=1] at (-10.4339,13.8984) {$B_h$};
\node [scale=1] at (-7.5868,12.8843) {${G_h}$};
\end{tikzpicture}
\qquad
\begin{tikzpicture}[scale=.5]

\draw[ line width=.7pt]  (-13,12) ellipse (3.5 and 2);

\node[fill=black, circle, scale=.2] (v22) at (-11.7313,11.5405) {};
\node [fill=black, circle, scale=.2](v10) at  (-11.1805,11.6144) {};
\node [fill=black, circle, scale=.2](v23) at (-10.6297,11.2319) {};
\node [fill=black, circle, scale=.2](v11) at (-11.8384,10.88) {};
\node [fill=black, circle, scale=.2](v24) at (-11.3182,10.6352) {};
\node [fill=black, circle, scale=.2](v5) at (-13.0471,13.6519) {};
\node[fill=black, circle, scale=.2] (v1) at (-14.3017,13.3306) {};
\node [fill=black, circle, scale=.2](v3) at (-13.9498,12.8256) {};
\node [fill=black, circle, scale=.2](v9) at (-12.481,12.8104) {};
\node[fill=black, circle, scale=.2] (v8) at (-11.7619,13.3306) {};
\node[fill=black, circle, scale=.2] (v2) at (-15.3268,11.9842) {};
\node[fill=black, circle, scale=.2] (v4) at (-14.9137,11.464) {};
\node[fill=black, circle, scale=.2] (v21) at (-14.7607,10.7449) {};
\node [fill=black, circle, scale=.2](v19) at (-15.9388,11.5864) {};
\node [fill=black, circle, scale=.2](v20) at (-15.5869,11.0662) {};
\node (v12) at (-16.8378,9.891) {};
\node (v13) at (-16.2317,9.5461) {};
\node (v14) at (-10.4738,8.8982) {};
\node (v15) at (-9.5855,9.4625) {};
\node (v16) at (-12.429,12.3502) {};
\node (v17) at (-13.519,12.5279) {};

\node at (-13,12) {};
\node (v18) at (-13.0549,12.2004) {};
\node[scale=1] at (-10.4339,13.8984) {$B_h^0$};

\draw [black,line width=.7pt] (v2.center) -- (v19.center) -- (v20.center) -- (v21.center);
\draw [black,line width=.7pt](v1.center) -- (v5.center) -- (v3.center) -- (v9.center) -- (v8.center) -- (v5.center);
\draw [black,line width=.7pt](v1.center) -- (v3.center) -- (v8.center) -- (v1.center) -- (v9.center) -- (v5.center);
\draw [black,line width=.7pt](v19.center) -- (v4.center) -- (v20.center) -- (v2.center) -- (v21.center) -- (v19.center) (v2.center) -- (v4.center) -- (v21.center);
\draw [black,line width=.7pt] (v20.center) edge (v4.center);
\draw [black,line width=.7pt ](v22.center) -- (v11.center) -- (v23.center) (v10.center) -- (v24.center);

\draw [black,line width=.7pt](v22.center) -- (v10.center) -- (v23.center) -- (v24.center) -- (v11.center) -- (v10.center) (v22.center) -- (v24.center) (v22.center) -- (v23.center);
\draw [white,line width=.7pt] (-8,15.5) node (v25.center){} -- (-9,15.5) -- (-6,6.5) -- (-10,6.5) --  (-8,15.5);

\end{tikzpicture}
  \caption{Local shortcut construction. On the left, $T$ is solid and dotted red. Dotted edges are edges absent from the partial $k$-cliques. On the right is $B_h^0$ for the $B_h$ on the left.}
  \label{figure:LocalShortcuts}
\end{figure}

\hdr{Local Shortcuts} See~\Cref{figure:LocalShortcuts}.
Let $h$ be an arbitrary bag, we apply the following argument to all of them. We now focus on the local shortcut within $B_h$. Let $T_h^1 := T \cap B_h$ be the forest when we look at $B_h$ in isolation (note that the tree $T$ can become disconnected). We will repair $T_h^1$ in the next paragraph.

Let $B_h^0\in\m F$ be the original bag of $B_h$, which is $B_h$ with all partial $k$-cliques involved in the clique-sum completed to full $k$-cliques (see~\Cref{figure:LocalShortcuts}). In particular, $V(B_h) = V(B_h^0)$. We emphasize that $B_h^0 \in \m F$ by the definition of partial-cliques.

In order to find a tree-restricted shortcut on $B_h$, we have to define the tree. The forest $T_h^1 := T \cap B_h^0$ might be disconnected, so we have to repair it. First, we define a \textbf{path contraction} operation between two vertices $s, t \in V(B_h^0)$. Consider the unique path $s = u_0, u_1, \ldots, u_{\ast} = t$ between them in $T$. Delete any vertex $u_i \not \in V(B_h)$ and one is left with a valid path in $B_h^0$ between $s$ and $t$. Note that the contracted path is a graph minor of $T$.

We form the repaired tree $T_h^2$ in the following way: for every two $s, t \in V(B_h^0)$, take the path contraction between them and union it into $T_h^2$. It is clear that (1) $T_h^2$ is a subgraph of $B_h^0$, in fact, it is a spanning tree of $B_h^0$, (2) $T_h^1$ is a subgraph of $T_h^2$, and (3) $T_h^2$ is a minor of $T$. The last property implies that $T_h^2$ is connected and that its diameter is at most $d_T$. Also, note that the same argument shows that for any part $P$, its restriction $B_h^0[P]$ is also connected since we can contract any path inside $P$ and the resulting path is still in $B_h^0$ and contains only vertices in $P$---the only unimportant difference being that this path might be on $T$.

Next, construct a $T_h^2$-restricted shortcut, discard all edges in $T_h^2 \setminus T = T_h^2 \setminus T_h^1$, and discard all edges contained in $C_f$, where $f$ is the parent $\dectree$-edge of $h$. The resulting assignment is the local shortcut of $B_h$.

The congestion of the local shortcut is $c_{\m F}(d_T)$. Fix an edge $e \in E(G)$, and note that it is only locally assigned in the bag $h_e$ (due to discarding (2)). But the local congestion of $h_e$ is $c_{\m F}(d_T)$, as claimed. The total congestion is at most the sum of the local and global one, hence it is at most $k d_{\dectree} + c_{\m F}(d_T)$.

\hdr{Bounding the Block Parameter}
With all shortcut edges established, we now upper bound the block parameter for each part $P \subseteq V(G)$. Remember that $T$ is, arbitrarily, rooted. We will bound the number of nodes $v \in V(G)$ that are roots of block components. Note that $v \in B_{h_P}$ since otherwise the global shortcut assigns the $T$-parent edge of $v$ to $P$. But in the lowest common ancestor $B_{h_P}$, $v$ can be a block root only if either (a) it is a vertex in $C_f$, where $f$ is the parent $\dectree$-edge of $h_P$, or (b) it is a block root of a local shortcut inside $B_{h_P}$. Summing up the contributions of these two cases, the total number of block roots, and therefore block components, can be at most $k + b_{\m F}(d_T)$.

\end{proof}

To improve the $d_{\dectree}$ factor in the congestion and prove the main result of this section, we compress the decomposition tree $\dectree$ to reduce its depth to $O(\log^2n)$, in a similar way to the compression scheme in \cite{bodlaender1988nc} for treewidth decompositions.

\CliqueSum*

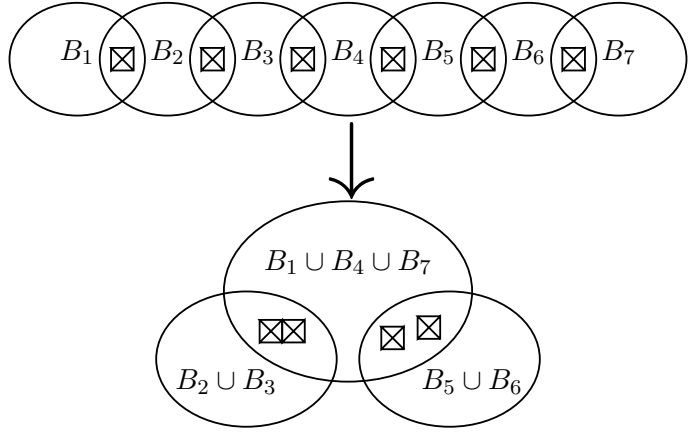
\begin{figure}
  \centering
  \begin{tikzpicture}[scale=.3]
\node at (11,3.8) {$B_{7}$};
\draw [black,line width=.7] (11,3.5) ellipse (3 and 2.5);
\node at (7,3.8) {$B_{6}$};
\draw [black,line width=.7] (7,3.5) ellipse (3 and 2.5);
\draw [black,line width=.7](8.5,4)  -- (8.5,3)  -- (9.5,4) -- (9.5,3)  -- (8.5,3) (9.5,4) -- (8.5,4) --  (9.5,3);
\node at (3,3.8) {$B_{5}$};
\draw [black,line width=.7] (3,3.5) ellipse (3 and 2.5);
\draw [black,line width=.7](4.5,4)  -- (4.5,3)  -- (5.5,4) -- (5.5,3)  -- (4.5,3) (5.5,4) -- (4.5,4) --  (5.5,3);
\node at (-1,3.8) {$B_{4}$};
\draw [black,line width=.7] (-1,3.5) ellipse (3 and 2.5);
\draw [black,line width=.7](0.5,4)  -- (0.5,3)  -- (1.5,4) -- (1.5,3)  -- (0.5,3) (1.5,4) -- (0.5,4) --  (1.5,3);
\node at (-5,3.8) {$B_{3}$};
\draw [black,line width=.7] (-5,3.5) ellipse (3 and 2.5);
\draw [black,line width=.7](-3.5,4)  -- (-3.5,3)  -- (-2.5,4) -- (-2.5,3)  -- (-3.5,3) (-2.5,4) -- (-3.5,4) --  (-2.5,3);
\node at (-9,3.8) {$B_{2}$};
\draw [black,line width=.7] (-9,3.5) ellipse (3 and 2.5);
\draw [black,line width=.7](-7.5,4)  -- (-7.5,3)  -- (-6.5,4) -- (-6.5,3)  -- (-7.5,3) (-6.5,4) -- (-7.5,4) --  (-6.5,3);
\node at (-13,3.8) {$B_{1}$};
\draw [black,line width=.7] (-13,3.5) ellipse (3 and 2.5);
\draw [black,line width=.7](-11.5,4)  -- (-11.5,3)  -- (-10.5,4) -- (-10.5,3)  -- (-11.5,3) (-10.5,4) -- (-11.5,4) --  (-10.5,3);

\draw [black,line width=.7](-4.9159,-8.0613)  -- (-4.9159,-9.0613)  -- (-3.9159,-8.0613) -- (-3.9159,-9.0613)  -- (-4.9159,-9.0613) (-3.9159,-8.0613) -- (-4.9159,-8.0613) --  (-3.9159,-9.0613);

\draw [black,line width=.7](0.476,-8.3619)  -- (0.476,-9.3619)  -- (1.476,-8.3619) -- (1.476,-9.3619)  -- (0.476,-9.3619) (1.476,-8.3619) -- (0.476,-8.3619) --  (1.476,-9.3619);

\draw [black,line width=.7](2.0852,-7.9015)  -- (2.0852,-8.9015)  -- (3.0852,-7.9015) -- (3.0852,-8.9015)  -- (2.0852,-8.9015) (3.0852,-7.9015) -- (2.0852,-7.9015) --  (3.0852,-8.9015);

\draw [black,line width=.7](-3.9159,-8.0613)  -- (-3.9159,-9.0613)  -- (-2.9159,-8.0613) -- (-2.9159,-9.0613)  -- (-3.9159,-9.0613) (-2.9159,-8.0613) -- (-3.9159,-8.0613) --  (-2.9159,-9.0613);

\draw [black,line width=.7] (-5.5,-9.8) ellipse (4 and 3);
\draw [black,line width=.7] (3.5,-9.8) ellipse (4 and 3);
\draw [black,line width=.7] (-1,-6.8) ellipse (5.5 and 4);
\node[scale=3,rotate=-90] at (-1.1,-1) {$  \rightarrow $};
\node at (-1,-5.5) {$B_1\cup B_4\cup B_7$};
\node at (-6.4,-10.8) {$B_2\cup B_3$};
\node at (4.5,-10.8) {$B_5\cup B_6$};
\end{tikzpicture}
  \caption{Compressing a $k$-clique-sum decomposition tree with high depth.}
  \label{figure:HLD}
\end{figure}
\BP
Let $\dectree$ be a $k$-clique-sum decomposition tree of $G$. To motivate the main proof, we first consider the case when $\dectree$ is a single path from root to leaf. This case will directly help in the general case, in which we apply \textit{heavy-light decomposition} to the tree, breaking it up into chains, and then treat each chain as a single path; we will present this general case next.

\hdr{Case When $\dectree$ Is a Path} Assume that $\dectree$ is a rooted path with bags $B_1,\lds,B_{d_{\dectree}}$, in that order. We recursively construct a balanced binary decomposition tree $\dectree'$ as follows.
\BE
 \im Group the bags $B_1,B_{\lc d_{\dectree}/2\rc},B_{d_{\dectree}}$ into a single bag $B_r$.
\im Recursively solve the paths $B_2,\lds,B_{\lc d_{\dectree}/2\rc-1}$ and $B_{\lc d_{\dectree}/2\rc+1},\lds,B_{d_{\dectree}}$. 
\im Attach the two resulting trees as subtrees of $B_r$ (see \Cref{figure:HLD}).
\EE
We call this operation \textbf{folding} a path.

 Call the new decomposition tree $\dectree'$; it is almost a $k$-clique-sum decomposition tree, with one exception: an edge may no longer a partial $k$-clique, but a union of two partial $k$-cliques. We call such edges \textbf{double edges}. Note that, while we can add edges within each of the two partial $k$-cliques and keep the graph in the family $\m F$, we cannot add edges between a vertex in one partial $k$-clique and a vertex in the other. Hence, we cannot simply treat the union of two partial $k$-cliques as a single partial $2k$-clique. However, a bag $B_i$ can have at most two children connected by double edges.

 Using the terminology of the above proof, let $B_h^0$ still be the bag $B_h$ with all partial cliques filled in with edges (the union of two cliques in a double edge will not have edges between them). The only difference this incurs in the proof is the following: in the global shortcut, partial cliques on the edge of $\dectree$ can now contain $2k$ vertices instead of $k$, doubling the congestion; and, in the local shortcut, a part restricted to to a bag $B_h^0[P]$ might not be connected anymore. However, we claim that it consists of at most $O(1)$ connected components: for each connected component we find a ``representative vertex'' in that component as follows. If (1) the component touches a partial clique in a double edge to a child, then the representative is the lowest numbered vertex in such a partial clique, and otherwise (2) we pick any vertex in the component. Now there will be at most $O(1)$ different representatives, thereby finishing the claim since no two different components can have the same representative. One can see this by arguing if (1) a part touches a partial clique in a double edge to a child, the it has at most 4 possibilities; otherwise (2) the part is already connected via the previous proof.

 We construct local shortcuts considering connected components of the parts as separate (sub)parts and union the assignment in the end. This only decreases the congestion, and increases the block by a multiplicative $O(1)$ to a total of $2k + O(b_{\m F}(d_T))$.

We now discuss the general case, when $\dectree$ is an arbitrary tree. The main steps of the proof are as follows. First, we compute a \textit{heavy-light decomposition}~\cite{harel1984fast} of $\dectree$. Then, we fold every chain in the heavy-light decomposition the same way we fold a single path, so that the resulting tree decomposition has depth $O(\log^2n)$.

\hdr{Heavy-Light Decomposition}
The heavy-light decomposition is a decomposition of any rooted tree into vertex-disjoint paths, called \textbf{heavy chains}, such that any path from the root to a leaf changes at most $O(\log n)$ heavy chains, where $n$ is the number of vertices in the tree. The decomposition is simple: for each non-leaf vertex of the tree, connect it to the child vertex with the largest number of vertices in its own subtree. On any path from root to leaf, if traveling from vertex $u$ to vertex $v$ changes heavy chains, then vertex $u$ has at least twice as many vertices in its subtree than does $v$; such an event can only occur $\log_2n$ times along the path.

\hdr{Folding a Chain} Once we compute the heavy-light decomposition, we partition the vertices of $\dectree$ into heavy chains, and then fold each chain independently.
Then, we connect the resulting binary trees in the following natural way: if the root of chain $C_1$ is a child of some vertex $v$, then we connect the root of the binary tree of $C_1$ to $v$. Note that this is not a double edge. We get a rooted tree $\dectree'$ of depth $O(\log^2 n)$ with the following key property: while every vertex in the new decomposition tree can have many children, it has at most two children connected via double edges. Therefore, the same argument for double edges in the single path case also applies here.  With the depth of $\dectree'$ reduced to $O(\log^2n)$, the result follows.
\EP
} 

\subsection{Shortcuts in Almost Embeddable Graphs}
\label{section:AlmostEmbeddable}

\shortOnly{
The full version of this paper proves \Cref{thm:almostEmbeddable}. In particular, we prove that $k$-almost-embeddable graphs admit good shortcuts. Recall that these graphs have bounded genus with an additional constant number of apices and vortices of constant depth added.

However, due to space constraints, we will only prove a simple sub-case: we assume that there are no apices in the graph. Unfortunately, dealing with both apices and vortices constitutes a majority of the technical difficulties and novel approaches of this paper. Still, the non-apex graphs provide a good warm-up.
}

\fullOnly{
In this section, we prove \Cref{thm:almostEmbeddable}. In particular, we prove that $k$-almost-embeddable graphs admit good shortcuts. Recall that these graphs have bounded genus with an additional constant number of apices and vortices of constant depth added.
}

\subsubsection{Warm-up: Non-Apex Graphs}\label{section:NonApex}

As a warm-up, we disregard apices and only consider graphs of bounded genus with vortices, i.e., the ``\textbf{Genus+Vortex}'' graphs. We establish tree-restricted shortcuts with block parameter $O((g+1)kD)$ and congestion $O((g+1)kD\log n)$ for graphs of genus $g$ with a $k$-vortex included. We first show that such a graph must have treewidth at most $O((g+1)kD)$, and then use the treewidth-based shortcut construction of \cite{haeupler2016near}. We note that this lemma is not novel, it is a simple consequence of the work by Dujmovic, Morin and Wood~\cite{dujmovic2017layered}, but we chose to include it because it illustrates how to deal with vortices.

At this point, we introduce our notation for treewidth decompositions. A \textit{treewidth decomposition} of a graph $G$ is a tree $\dectree$ whose vertices, called \textit{bags}, are subsets of $V(G)$. The tree $\dectree$ satisfies three properties: (i) the union of vertices over all bags equals $V(G)$; (ii) for each $v\in V$, the set of bags containing $v$ is connected in $\dectree$; (iii) for each edge $(u,v)\in E(G)$, there is a bag containing both $u$ and $v$. The \textit{treewidth} of a graph $G$ is the minimum $k$ such that there exists a tree decomposition $\dectree$ of $G$ whose bag sizes are all at most $k+1$.

\begin{lemma}\label{lemma:GenusVortex}
A graph $G$ of diameter $D$ and genus $g$ with a single vortex of depth $k$ has treewidth $O((g+1)kD)$.
\end{lemma}

\BP
First, we transform $G$ into a graph $G'$ of genus $g$ and diameter at most $D + 1$ as follows: remove all the vertices inside the vortex, and add a single vertex $r$ in the vortex face with an edge to all vertices on the vortex boundary. Since pairwise distances between vertices on the boundary do not increase by more than 1, the diameter of $G'$ is at most $D + 1$. 

Eppstein~\cite{eppstein2000diameter} proves that graphs of genus $g$ have treewidth $O((g+1)D)$. Therefore, there exists a tree decomposition $\dectree'$ of $G'$ with bag size $O((g+1)D)$. Remove $r$ from $\dectree'$. To add the vortex back in, first take a vortex decomposition $\m P$. Then, for each vertex $v$  inside the vortex that was removed, add $v$ to every bag in $\dectree'$ that intersects $\m P(v)$, i.e., contains a boundary vertex on the corresponding arc of $v$. It remains to prove that the resulting tree decomposition $\dectree$ is valid and has bag size $O((g+1)kD)$.

To show the former, fix a vertex $v$ inside the vortex. Since the neighboring boundary vertices in $\m P(v)$ are connected by edges, there exists a common bag between every two neighboring boundary vertices. Therefore, the entire set of bags containing $v$ is connected. In addition, $v$ shares a common bag with any boundary vertex in $\m P(v)$, as well as any other vertex $v'$ in the vortex with $\m P(v)\cap\m P(v')\ne\emptyset$. It follows that for each edge incident to $v$, there exists a bag containing both of its endpoints.

Finally, since the vortex decomposition $\m P$ has depth at most $k$, each vertex on the boundary is responsible for at most $k$ new vertices in its bags. Since each bag has at most $O((g+1)D)$ boundary vertices, the new bag size is $O((g+1)kD)$.
\EP

This proof easily generalizes to the case when $G$ has $\ell$ vortices, each of depth $k$.

\begin{lemma}\label{lemma:VortexTreewidth}
A graph $G$ of diameter $D$ and genus $g$ with $\ell$ vortices of depth $k$ has treewidth $O((g+1)k\ell D)$.
\end{lemma}

\shortOnly{\vspace{-2mm}}
Finally, applying the treewidth-based shortcut construction, namely \Cref{theorem:shortcut-existence-on-bounded-treewidth}, gives the desired result.

\begin{theorem}
A genus $g$ and diameter $D$ graph with $\ell$ vortices of depth $k$ has tree-restricted shortcuts with congestion $O((g+1)k\ell D\logn)$ and block parameter $O((g+1)k\ell D)$.
\end{theorem}

\fullOnly{
In particular, since planar graphs have genus $0$, we get the following corollary:
\begin{corollary}
A diameter $D$ planar graph with $\ell$ vortices of depth $k$ has tree-restricted shortcuts with congestion $O(k\ell D\logn)$ and block parameter $O(k\ell D)$.
\end{corollary}
}

\fullOnly{
\subsubsection{Apex Graphs}

In this section, we add apices to $(0,g,k,l)$-almost-embeddable (``Genus+Vortex'') graphs. At first glance, the addition of an apex to a graph might seem trivial, since the graph only changes by one vertex, and using that vertex can only make the shortcuts better. However, notice that the diameter of the graph can shrink arbitrarily with the addition of an apex, and our shortcuts on the apex graph must be competitive with the new diameter. Hence, we need ideas beyond our shortcut constructions for the graph without the apex. For a simple example, in a cycle graph, shortcuts with quality $\Theta(n)$ are considered good. However, by adding a single central vertex, we can transform the graph into the wheel graph where ``good'' shortcuts should have quality $\Theta(1)$. While good shortcuts actually do exist in the wheel graph, there are examples of graphs with good shortcuts where adding a apex makes good shortcuts impossible.

To streamline our arguments for $(q,g,k,l)$-almost-embeddable graphs (``Apex+Genus+Vortex'') graphs, we will define a couple of intermediate properties which do not depend on the graph topology. More precisely, we will define the notions of \textbf{$\beta$-cell-assignment} and \textbf{$s$-combinatorial gates}. On a very high level, We will show that:
\begin{enumerate}
\im A Genus+Vortex graph has an $s$-combinatorial gate, for an appropriately chosen $s$. (\Cref{section:SatisfyingThe} and the Appendix)
\im Graphs with $s$-combinatorial gate are $\beta$-cell-assignable, for appropriately chosen $\beta$ and some technical stipulations. (\Cref{section:TheCellsparse})
\im Graphs that are $\beta$-cell-assignable and each cell locally admits good tree-restricted shortcuts also globally admit good tree-restricted shortcuts, barring various technicalities. (\Cref{section:FromCellsparse})
\end{enumerate}

In each part, we separately prove the statements with Genus+Vortex graphs replaced by planar graphs. It is recommended that the reader, in their first reading, focus only on the lemmas regarding planar graphs with a single apex, namely Lemmas \ref{lemma:SufficientCondition}, \ref{lemma:combinatorialToCellsparse}, \ref{lemma:PlanarCombinatorial}, and \ref{lemma:PlanarShortcut}.

\subsubsection{Cell Partitions, $\beta$-Cell-Assignment and $s$-Combinatorial Gate}\label{section:TheCellsparse}

In this section, we first introduce the notions of ``cell partitions'', ``$\beta$-cell-assignment'' and ``$s$-combinatorial gates''. Second, we prove that the second property implies the first.

\begin{definition}
  A \textbf{cell partition} of $G$ is simply a partition of $V_G$ into disjoint, connected components with a small diameter, called the \textbf{cells}.
\end{definition}
Note that the diameter condition is the only thing differentiating it from the definition of parts. It is helpful to think of cells as low-diameter components, whereas parts may be long and skinny. A canonical example for a cell partition is the following. Given an apex graph of diameter $D$, remove the apex and start a concurrent BFS from each node adjacent to the removed apex. Each node in the graph (except the apex) gets assigned to exactly one BFS component. We call such BFS components cells. For most of this section, we will ignore any extra property that a cell partition might have and assuming nothing besides them being disjoint, connected and having a controlled diameter.

A graph is cell-assignable if we can relate its cells and parts in a way that no cell is assigned to too many parts and parts are assigned to \textbf{almost all} intersecting cells.
\begin{definition}
  \label{def:cellsparse}
  A graph $G = (V_G, E_G)$ is $\beta$-cell-assignable if the following holds. For every valid family of parts $\mc P$ (as in \Cref{def:generalshortcut}) and every valid cell partition $\mc C$ of diameter $d$ there exists a relation $\mc R \subseteq \mc C \times \mc P$ with the following properties:
  \begin{itemize}
  \item[(i)] each part is in relation with \textbf{all} cells it intersects, \textbf{except} for at most 2 of them
  \item[(ii)] each cell is in relation with at most $\beta$ parts
  \end{itemize}
  Note: $\beta$ is a function of the cell diameter $d$.
\end{definition}

We will not prove directly that Genus+Vortex graphs are $\beta$-cell-assignable. Instead, we focus on a combinatorial property that we show implies cell-assignment. This property is called a ``combinatorial gate'' and it intuitively asserts that every two touching cells have a ``gate'' that covers all the edges between them. Furthermore, the boundary of such a gate is called a ``fence'' and its size should be controlled. The reader is encouraged to review \Cref{figure:planar-boundary} for a mental picture of combinatorial gates on a planar graph.

\BD
For a subset of vertices $S\s V$, define the $\partial S$ to be the set of vertices in $S$ on the boundary of $S$, i.e., the vertices in $S$ whose neighborhoods intersect $V\bs S$.
\ED

\BD \label{defn:CombGate}
Let $G=(V,E)$ be a graph from a family $\m F$, and let $\m C$ be a partition of $G$ into cells. We define a $s$-\textbf{combinatorial gate} to be a collection $\m S = \{ (F_i, S_i) \}_i$ where $F \subseteq V$ are called \textbf{fences}, $S \subseteq V$ are \textbf{gates}, and the following properties hold:
\BE
\im Fences are a subset of their corresponding gates. I.e., $F \subseteq S$ for all $(F, S) \in \m S$.
\im The boundary of a gate are included in its fence. I.e., $\partial(S) \subseteq F$ for all $(F, S) \in \m S$.
\im Each edge $\{a, b\} \in E$ whose endpoints are in different cells must be covered by some gate. I.e., $a \in S \land b \in S$ for some gate $S$.
\im Each gate $S$ intersects at most two cells in $\m C$.
\im The non-fence vertices of the gates are disjoint. I.e., for every $v \in V$ there is at most one $(F_i, S_i) \in \m S$ s.t. $v \in S_i \bs F_i$.
\im The average size of fences compared to the number of cells is at most $s$. I.e., $\sum_{(F, S) \in \m S} |F| \le s |\m C|$. 
\EE
\ED


Since this condition is entirely combinatorial, the proofs that imply $\beta$-cell-assignment are also combinatorial. Therefore, these results are self-contained and disregard any possible structure in the graph, for example, planarity. Next, we prove that the $s$-combinatorial boundary implies $\beta$-cell-assignment via the following two lemmas.

\begin{lemma}\label{lemma:SufficientCondition}
Suppose a graph $G$ with cell partition $\m C$ has an $s$-combinatorial gate $\m S$. Then, for any collection of parts $\m P$, either  there exists a part intersecting at most two cells, or there exists a cell intersecting at most $2s$ parts.
\end{lemma}

\BP
Let $\m I \subseteq \m P \times \m C$ be the set of pairs $(P, C)$ s.t. $P \cap C \neq \emptyset$. Let the ``degree of a part $P$'' be the cardinality $deg(P) := |\{ C \mid (P, C) \in \m I \}|$, and similarly define the ``degree of a cell $C$'' $deg(C)$.

We are done if there is a part with degree at most 2, hence we can assume $deg(P) \ge 3, \forall P \in \m C$. Fix a part $P$ and define $B_P = \{ i \mid P \cap F_i \neq \emptyset, F_i \text{ is a fence}\}$ be the indices of fences it intersects. Then $\sum_{P \in \m P} |B_P| \le \sum_i |F_i|$ since every fence vertex can be contained in at most 1 part.

Furthermore, fix a part $P$; we claim that $deg(P) \le |B_P|+1 \le 2|B_P|$. This paragraph proves the first inequality, the second being trivial. Since $P$ intersects $deg(P)$ many different cells, there must be $deg(P)-1$ edges whose endpoints are in different cells and are both in $P$; even more, the unordered pairs of cells the edges connect are different among edges. Property (3) of the combinatorial gate definition implies that all of these edges must be inside some gate $S_i$. However, it is impossible that $P \subseteq S_i$, otherwise property (4) would imply $deg(P) \le 2$. Therefore, $P$ must contain a vertex $\partial(S_i)$, which is also included in the fence $\partial(S_i) \subseteq F_i$ by property (2). We conclude that $i \in B_P$. Moreover, the $i$'s corresponding to different edges are distinct since the unordered pair of cells they are connecting is different.

We are now ready to prove the Lemma via the following claim: $|\m I| = \sum_{P \in \m P} \le 2 \sum_{P \in \m P} |B_P| \le 2s |\m C|$. Hence $\frac{|\m I|}{|\m C|} \le 2s$, implying there exists a cell with degree at most $2s$ by the pigeonhole principle.
\EP

\begin{lemma}\label{lemma:combinatorialToCellsparse}
Let $\m F$ be a family of graphs that is closed under taking minors. Suppose that there is a function $s(d):\N\to\N$ such that every graph $G\in\m F$ satisfies the following property:
\BI \im
  If $G$ has a cell partition of diameter $d$, then there exists an $s(d)$-combinatorial gate $\m S$ of subsets of $V(G)$.
\EI
Then, every graph $G\in\m F$ with a cell partition of diameter $d$ is $2s(d)$-cell-assignable.
\end{lemma}
\begin{proof}
  Fix a graph $G\in\m F$, cells $\m C=\{C_1,\lds,C_{|\m C|}\}$ of diameter $d$, and parts $\m P=\{P_1,\lds,P_{|\m P|}\}$. We construct an assignment $\m R$ following \Cref{def:cellsparse}. Assume that $G$ is connected; otherwise, we can repeat the argument below on each connected component of $G$. We proceed by induction on $|\m C|+|\m P|$, with the base case $|\m C|=1$ or $|\m P|=1$ being trivial.

  Suppose that $|\m C|>1$ and $|\m P|>1$. By \Cref{lemma:SufficientCondition}, either there is a part $P\in\m P$ intersecting at most two cells, or there exists a cell intersecting at most $2s(d)$ parts. In the former case, we do not assign any cell to $P$ in $\m R$ and proceed by induction on the instance $(G,\m C,\m P\bs P)$.

  In the latter case, we find a cell $C\in\m C$ intersecting at most $2s(d)$ parts and $\mathcal{R}$-assign $C$ to all parts it intersects. Then, iteratively remove $C$ from the graph by contractions. Repeatedly pick any remaining $v \in C$. If $v$ belongs to some part $P\in\m P$ and has a neighbor in $P$, then contract $v$ along any incident edge that has both of its endpoints in $P$. Note that, by the connectedness of $P$, $v \in P$ must have a neighbor in $P$ unless $P = \{ v \}$. Otherwise, contract a vertex $v$ it along any incident edge.

  Let $G'$ be the resulting graph, let $\m P'=\{P \setminus C \mid P\in\m P\}$ be the new partition, and  $\m C'=\m C\bs C$ be the remaining cells. Note that, by our edge contraction scheme, all parts in $\m P'$ remain connected in $G'$, all remaining cells remain connected, and incidences between the remaining cells and parts are unchanged. In addition, since the graph family $\m F$ is closed under edge contraction, $G'\in\m F$. We apply induction on the instance $(G', \m C', \m P')$ and union the resulting relation $\mathcal{R}'$ with the assignments made in the current iteration.
\end{proof}

While \Cref{lemma:combinatorialToCellsparse} works well for planar graphs that are closed under taking minors, Genus+Vortex graphs do not have that property due to the existence of a bounded number of vortices. In particular, if one contracts an edge inside the vortex, the resulting graph is not Genus+Vortex. Therefore, we will deal with cells touching vortices as ``special cells'' that are not allowed to be contracted.

\begin{lemma}\label{lemma:planar-apex-relation}
Let $\m F$ be a family of graphs, not necessarily closed under taking minors. Suppose that there is a function $s(d):\N\to\N$ such that every graph $G\in\m F$ satisfies the following property: 
\BI \im
  If $G$ has a cell partition of diameter $d$, then there exists an $s(d)$-combinatorial collection $\m S$ of subsets of $V(G)$.
\EI
Consider a graph $G\in\m F$ with a cell partition into two types of cells---normal cells and $\ell$ special cells---both of diameter $d$. Let $E^*$ denote the set of edges in special cells. Assume that any graph $G'$ obtained by deleting vertices and contracting edges outside of special cells is still in $\m F$. Then, $G$ is $2\ell s(d)$-cell-assignable with respect to a cell partition of only the normal cells.
\end{lemma}
\BP
Fix a graph $G\in\m F$, normal cells $\m C^0$, special cells $\m C^*$, and parts $\m P=\{P_1,\lds,P_{|\m P|}\}$.
Similarly to \Cref{lemma:combinatorialToCellsparse}. we proceed by induction on $|\m C^0|+|\m P|$, with the same base case being trivial.

Suppose that $|\m C^0|>1$ and $|\m P|>1$. If there is a part $P\in\m P$ intersecting at most two cells, then we proceed as in \Cref{lemma:combinatorialToCellsparse}. Otherwise, as in the proof of \Cref{lemma:SufficientCondition}, we show that the total number of pairs $(P,C)$ where a part $P\in\m P$ intersects a cell $C\in\m C\cup\m C^*$ is at most $2s(d)|\m C^0\cup\m C^*|$. Since $|\m C^0|\ge|\m C^0\cup\m C^*|-\ell\ge|\m C^0\cup\m C^*|/\ell$, the number of $(P,C)$ is at most $2\ell s(d)|\m C^0|$, so there exists a normal cell intersecting at most $2\ell s(d)$ parts. The rest of the proof is identical to that in \Cref{lemma:combinatorialToCellsparse}, except we note that since we only remove vertices in a normal cell, the new graph $G'$ still satisfies the conditions in the lemma.
\EP

\subsubsection{Graphs with $s$-Combinatorial Gate Property}\label{section:SatisfyingThe}

In this section, we show that Genus+Vortex graphs satisfy the $s$-combinatorial property. We highlight our main ideas by proving the statement for planar graphs, and defer the rest to the Appendix.

\begin{lemma}\label{lemma:PlanarCombinatorial}
Let $G$ be a planar graph with a cell partition of diameter $d$.
Then, there is an $36d$-combinatorial gate $\m S$.
\end{lemma}
\begin{proof}

\begin{figure}
  \centering
  \includegraphics[scale=1]{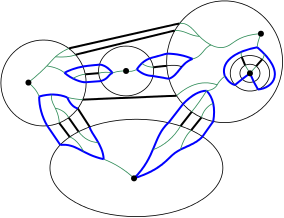}
  \qquad
  \includegraphics[scale=1]{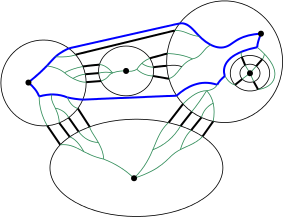}
  \caption{Graphical overview of our boundary construction. The black circles are the cells. Note that there is a cell completely contained inside another cell in the planar embedding. The green edges are the spanning trees $T_i$ as defined in the proof. The blue edges form our boundary construction. }
  \label{figure:planar-boundary}
\end{figure}

Fix a planar embedding of $G$ in the planar region $\R^2$. Define an auxiliary graph $A$ formed by contracting each cell into a single vertex, then removing parallel edges. In other words, two vertices in $A$ are adjacent iff their corresponding cells are connected by an edge; we call two such cells ``adjacent''. Our goal is to, for each pair of adjacent cells, define a closed loop that separates the planar embedding in a laminar way\footnote{A family of sets is laminar when any two members are either disjoint or one is a subset of another.}. \Cref{figure:planar-boundary} gives a graphical overview of our boundary construction, which we make more precise below. We note for later that the planarity of $A$ implies $|E(A)| \le 3|V(A)| - 6 = 3|\m C| - 6$.

For each cell $C_i$, define $T_i$ to be a spanning tree of $C_i$ with diameter at most $d$. For any two cells $C_i,C_j\in\m C$, define the set of ``$(C_i,C_j)$-inter-cell edges'' as those in $F$ that connect the two cells. Given two $(C_i,C_j)$-inter-cell edges $e_u:=(u_i,u_j)$ and $e_v:=(v_i,v_j)$, define the ``cycle along $e_u$ and $e_v$'' to be the union of edge $e_u$, edge $e_v$, the path along $T_i$ from $u_i$ to $v_i$, and the path along $T_j$ from $u_j$ to $v_j$. We denote the cycle by $cyc(e_u, e_v)$. Note that every such cycle has at most $4d + 2$ vertices.  Our next step is to find two special $(C_i,C_j)$-inter-cell edges $e_L$ and $e_R$ such that if we consider the cycle along $e_L$ and $e_R$, the boundary and interior of this loop, i.e., the set of points in the plane enclosed by the cycle, contain all $(C_i,C_j)$-inter-cell edges. We call $e_L$ and $e_R$ the ``extremal edges'' between $C_i$ and $C_j$.
\begin{figure}
  \centering
  \includegraphics[scale=1]{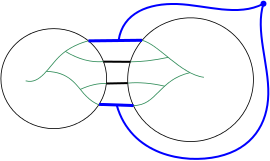}
  \qquad
  \includegraphics[scale=1]{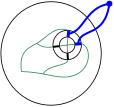}
  \caption{Definition of extremal edges between two different cells. $T_i$ and $T_j$ are shown in green. The extremal edges are the blue edges reachable from the outer face, as indicated by the paths in blue, while the other inter-cell edges are the black edges.}
  \label{figure:extremal}
\end{figure}

Intuitively, we choose the extremal edges to be the ``left-most'' and ``right-most'' $(C_i,C_j)$-inter-cell edges, but these can be formally defined as follows (see \Cref{figure:extremal}). Define a planar graph $T_{ij}$ to be the union of $T_i$, $T_j$, and all $(C_i, C_j)$-inter-cell edges. Note that $T_{ij} \subseteq G$ inherits the planar embedding from $G$. Draw a loop starting and ending in the outside face that encloses $T_i$ but lies outside of $T_j$; this can always be done because the trees are  disjoint. The first and last $(C_i,C_j)$-inter-cell edge intersected by the loop must lie on the outside face and form our extremal edges $e_L, e_R$\footnote{If there is only one $(C_i,C_j)$-inter-cell edge, then we set both $e_L$ and $e_R$ to that edge.}.


Take all pairs of adjacent cells $(C_i, C_j)$ and consider the cycle $K_{ij}\s G$ along their extremal edges $cyc(e_L, e_R)$. Let $\mathcal{K}$ be the set of all such cycles.

For a cycle $K \in \m K$ define the set of points in the embedding enclosed within the cycle as $reg(K)$. Note that $reg(K)$ is closed set. An important property of $\{ reg(K) \mid K \in \m K \}$ is that two $reg(K), reg(K')$ are either disjoint or one is a subset of another, i.e., they form a laminar family and the notions of minimal cycle and maximal cycle are well-defined inside the family.

We are ready to construct the combinatorial gate. Given a cycle $K_{ij} \in \m K$, let $own(K) := reg(K) \setminus \left( \cup_{reg(K') \subseteq reg(K)} int(reg(K')) \right)$, where $int(\cdot)$ is the topological interior. In other words, $reg(K)$ is the set of points in the embedding that are enclosed in the cycle, but are outside of the strict interiors of any other cycles $K' \in \m K$ enclosed in $K$. Let $S_{ij} \subseteq V(G)$ be the set of vertices $v \in C_i \cup C_j$ where the corresponding point in the embedding $p_v \in \mathbb{R}^2$ is in $own(K)$. Similarly, let $F_{ij} \subseteq V(G)$ be the set of vertices $v \in C_i \cup C_j$ which are not in the topological interior of $own(K_{ij})$ (again, in terms of the embedding). The combinatorial gate is then defined as $\mc S := \{ (F_{ij}, S_{ij}) \mid K_{ij} \in \mc K \}$.

Property (1), i.e., $F_{ij} \subseteq S_{ij}$, and Property (4), i.e., $S_{ij} \subseteq C_i \cap C_j$, trivially follow from the definition. Property (5): the laminarity of $\{ reg(K) \mid K \in \m K \}$ implies that $\{ own(K) \mid K \in \m K \}$ can only share a boundary, hence their interiors are disjoint; therefore, non-fence vertices have at most one $own(K)$ region they are contained in, implying the Property.

Property (3): fix any $(C_i, C_j)$-inter-cell edge $e$. By construction, $p_e \subseteq reg(K_{ij})$, where $p_e \subseteq \mathbb{R}^2$ is the set of points the edge corresponds to in the embedding. We also claim that $p_e \subseteq own(K_{ij})$, which would imply the Property. Assume this is not true, then $p_e \subseteq reg(K_{i'j'})$ for some $K_{i'j'} \neq K_{ij}$ such that $reg(K_{i'j'}) \subseteq reg(K_{ij})$. We can assume without loss of generality that $i \not \in \{i', j'\}$. By planarity of the graph and the connectedness of $C_{i}$, the points in the embedding corresponding with $C_{i}$ would have to lie inside of $reg(K_{i'j'})$ since no edge of $C_i$ can cross $K_{i'j'} \subseteq C_{i'} \cup C_{j'}$. But that contradicts the assumption that $reg(K_{i'j'}) \subseteq reg(K_{ij})$, implying the Property.



Finally, we prove Property (6) with the parameter $s := 36|\m C|d$. Every fence vertex $v \in F_{ij}$ must either lie on $K_{ij}$ or on a \emph{maximal} cycle nested within $K_{ij}$. If $v \in K_{ij}$, we will charge it to $K_{ij}$; otherwise, if $v$ is on a maximal nested cycle $K' \in \m K$ inside of $K_{ij}$, we charge it to $K'$. Note that for any $K \in \m K$, only vertices on $K$ and the the unique enclosing cycle (if one exists) can charge to $K$. Therefore, the total number of vertices charged to any $K \in K$ is at most $2 \cdot \max_{K' \in \m K} |K'| \le 2 \cdot (4d + 2) \le 12d$, from before. The number of incident cells, $|\m K|$, is equal to $|E(A)| \le 3|\m C| - 6 \le 3|\m C|$; leading to $\sum_{(F, S) \in \m S} |F| \le 3|\m C| \cdot 12d \le 36|\m C|$.

\end{proof}

 The full result is rather technical and is proved in \Cref{sec:proof:gateInGenusVortex}.

\begin{restatable}{lemma}{GenusVortexCombinatorial}
  \label{lemma:GenusVortexCombinatorial}
Let $G$ be a genus-$g$ graph with (a possibly unbounded number of) vortices of depth $k$, and consider a cell partition of diameter $d$ such that no vortex is split between more than one cell.
        Then, there exists  an $O((g+1)kd)$-combinatorial gate of $G$. 
\end{restatable}

\subsubsection{Wrapping Up: From $\beta$-Cell-Assignment to Good Shortcuts}\label{section:FromCellsparse}

In this section, we finalize our proof for tree-restricted shortcuts in almost embeddable graphs. We do this by showing that if an $(0, g, k, l)$-almost-embeddable (``Genus+Vortex'') graph is $\beta$-cell-assignable for small enough parameter $\beta$, then the same graph with $q$ added apices admits good tree-restricted shortcuts. We first assume that the apex graph has exactly one apex, then establish a simple reduction from the multiple apices case. We begin with the same statement for $(1,0,0,0)$-almost-embeddable (``Apex+Planar'') graphs.

\begin{lemma}\label{lemma:PlanarShortcut}
Let $G$ be a planar graph with a single apex and a diameter $d_T$ spanning tree $T$ of $G$. For a given set of parts, there exists a $T$-restricted shortcut with block parameter $O(\log d_T)$ and congestion $O(d_T\log d_T)$.
\end{lemma}
\BP
Let $x$ be the apex, and let $H:=G-x$ be the planar region. First, if a part $P$ contains $x$, we give $P$ the entire spanning tree; there can be at most one such part, so the congestion does not change asymptotically. From now on, assume that no part contains the apex $x$.

Consider removing the apex $x$, which breaks the tree $T$ into multiple connected subtrees in $H$. Note that each subtree has diameter at most $d_T$. For each subtree, let its vertices be a new cell $C$, and denote the subtree by $T[C]$. Since the family of planar graphs is closed under edge contraction, we can invoke \Cref{lemma:combinatorialToCellsparse}, so $H$ with cell partition $\m C$ is $\beta$-cell-assignable for $\beta(d):=O(d)$, so there is a corresponding relation $\m R\s\m C\times\m P$. As in \Cref{lem:CliqueSumBad}, we construct local and global shortcuts separately, using ideas from \Cref{section:NonApex} for the local shortcuts and the relation $\m R$ for the global shortcuts.

We first begin with global shortcuts. For each part $P\in\m P$ and every cell $C\in\m C$ assigned to $P$ in the relation $\m R$, assign all edges in $T[C]$ to part $P$, as well as the edge connecting the apex $x$ to $T_C$, called the \textbf{uplink}. Since every edge, with the exception of uplinks, belongs to one cell, and since every cell is in relation with at most $\beta(d_{T_C})$ parts, the congestion on each edge is $\beta(d_{T_C})\le\beta(d_T)$ from global shortcuts.

Next, we define local shortcuts by repeating the following for each $C \in \m C$ individually. (i) Take the graph $H$, and iteratively contract all edges in $E_H\bs E_{H[C]}$, i.e. all edges outside the graph induced by $C$. (They are contracted in the same way as in \Cref{lemma:combinatorialToCellsparse}, so that all parts $P\cap C$ remain connected in the resulting graph. Denote the graph by $H_C$. (ii) Next, construct a $b(d):=O(\log d)$ block parameter, $c(d):=O(d\log d)$ congestion $T_C$-restricted shortcut using~\Cref{theorem:shortcut-existence-on-bounded-genus} on $H_C$ with the parts $\{P\cap C:P\in\m P\}$. Note that the parts are still connected via the contractions. Add it as a local shortcut of the cell $C$. Each edge in $T$ has congestion $c(d_{T_C})\le c(d_T)$ from local shortcuts.

Finally, we argue about block parameter. For each part $P$, there are at most 2 cells intersecting $P$ but not in relation with $P$ in $\m R$. Each of these cells $C$ generates $b(d_{T_C})\le b(d_T)$ additional blocks, and together with the single block from the global shortcuts, gives a blocking parameter of $1+ 2 \cd b(d_{T})$.

Plugging in $\beta(d_T)=O(d_T)$, $b(d_T)=O(\log d_T)$, and $c(d_T)=O(d_T\log d_T)$, we get block parameter $1+2 \cd b(d_T)=O(\log d_T)$ and congestion $\beta(d_T)+c(d_T)=O(d_T\log d_T)$.
\EP

\begin{lemma}\label{lemma:GenusVortexShortcut}
Let $G$ be a  genus-$g$ graph with $\ell$ vortices of depth $k$ and a single apex and $T$ a spanning tree of $G$. For a given set of parts, there exists a $T$-restricted shortcut with block parameter $O((g+1)k\ell^2d_T)$ and congestion $O(k\ell^2d_T(g+\log n))$.
\end{lemma}
\BP
We proceed similarly as in \Cref{lemma:PlanarShortcut}, with different parameters $\beta$, $b$, and $c$. Let $G$ be the entire graph, let $x$ be the apex, and let $H:=G-x$ be the graph without the apex. Let $\m C$ be the cell partition as defined in \Cref{lemma:PlanarShortcut}. To obtain the actual partition $\m C'$ that we can use as a precondition in \Cref{lemma:GenusVortexCombinatorial}, we first start with $\m C$ and then, iteratively, for each vortex in $H$, merge all cells that intersect the vortex. Note that if a cell in $\m C'$ intersects a vortex, then it completely contains the vortex, and a cell may contain multiple vortices. We let all cells that contain, or equivalently, intersect, a vortex to be special, so that there are at most $\ell$ special cells. The remaining cells are normal cells. At this point, all normal cells have diameter $O(d_T)$, but special cells can have unbounded diameter due to the individual vortices. To remedy this issue, for each vortex, we create a \textit{star vertex} that connects to all boundary vertices of the vortex, and add it to the special cell containing this vortex. Doing so increases the depth of each vortex by at most $1$ and decreases the diameter of each special cell to $O(\ell d_T)$. Denote the normal cells by $C^0$ and the special cells by $\m C^*$. 

Define a graph family $\m F$ to be all genus-$g$ graphs with at most $\ell$ vortices of depth at most $k$, so that $H\in\m F$. Note that for each graph $F\in\m F$, contracting an edge outside of any vortex still leaves a graph in $\m F$. Since every normal cell is in the genus-$g$ region of the graph, any graph obtained by contracting edges in normal cells in $H$ is still in $\m F$. Therefore, we can apply \Cref{lemma:GenusVortexCombinatorial} with cell diameter $O(\ell d_T)$ to get an $O((g+1)k \ell d_T)$-combinatorial gate, and then apply \Cref{lemma:planar-apex-relation} to obtain a relation $\m R\s\m C'\times\m P$ with $\beta(d_T):=2\ell ((g+1)k\ell d_T)$. Note that our combinatorial gate applies to the graph with the extra star vertices added, not the original graph. However, only special cells get star vertices and the relation $\m R$ does not touch special cells, so $\m R$ is valid for the original graph.

The global shortcuts for normal cells are the same as those in \Cref{lemma:PlanarShortcut}, giving a congestion of at most $\beta(d_T)$. Note that there are no global shortcuts for special cells, since $\m R$ does not associate special cells. For local shortcuts in a normal cell $C\in\m C^0$, we define $H_C$ as follows: we first contract all edges in $E_H\bs \lp E_{H[C]}\cup \bigcup_{C'\in C^*}E_{H[C']}\rp$, i.e., all edges not inside $C$ or any special cell, in the same way as in \Cref{lemma:combinatorialToCellsparse}. We then contract the edges in $ \bigcup_{C'\in C^*}E_{H[C']} \bs E_{H[C]}$, i.e., the remaining edges not inside $C$, in the same way as in \Cref{lemma:combinatorialToCellsparse} to obtain $H_C$. The first set of contractions leaves $C$ with the special cells. With the star vertices added, this graph is a genus $g$ graph with $\ell$ vortices and diameter $O(\ell d_T)$, so by \Cref{lemma:VortexTreewidth}, it has treewidth $O((g+1)k\ell^2 d_T)$; disregarding the star vertices can only decrease the treewidth. The second set of contractions also cannot increase the treewidth, so $H_C$ also has treewidth $O((g+1)k\ell^2 d_T)$. Applying the treewidth-based shortcut construction from~\cite{haeupler2016near} gives block parameter $b(d_T):=O((g+1)k\ell^2 d_T)$ and congestion $c(d_T):=O((g+1)k\ell^2 d_T\log n)$.

For local shortcuts in special cells, we construct them in all special cells simultaneously. Let $C^*:=\bigcup_{C\in C^*}C$ be the union of all special cells. Define the tree $T^*:=T[x\cup C^*]$ to be the union of $C^*$ with all uplinks in $C^*$. Take the graph $H$ and contract all edges in $E_H\bs \bigcup_{C\in C^*}E_{H[C]}$, i.e., all edges not inside any special cell, in the same way as in \Cref{lemma:combinatorialToCellsparse}, obtaining a genus $g$ graph with $\ell$ vortices and diameter $O(\ell d_T)$; This graph has treewidth $O((g+1)k\ell^2d_T)$. Add the apex $x$ back to $H$ and connect it to its neighbors in $G$ that are vertices of $H$, which increases the treewidth by at most $1$. The resulting graph is spanned by $T^*$, so it has shortcuts with block parameter $b(d_T)$ and congestion $c(d_T)$, as defined above.

Finally, we argue about block parameter. For each part $P$, there are at most 2 normal cells intersecting $P$ but not in relation with $P$ in $\m R$, and at most $\ell$ special cells. Each of these cells generates $b(d_T)$ additional blocks, and together with the single block from the global shortcuts, gives a block parameter of $1+(2+\ell) \cd b(d_{T})$.

Plugging in $\beta(d_T)=O((g+1)k\ell^2d_T)$, $b(d_T)=O((g+1)k\ell^2 d_T)$, and $c(d_T)=O(k\ell^2 d_T\log n)$, we get block parameter $1+2\cd b(d_T)=O((g+1)k\ell^2d_T)$ and congestion $\beta(d_T)+c(d_T)=O(k\ell^2d_T(g+\log n))$.
\EP

We finally prove the main theorem of this section, with multiple apices.
\BuildingBlocks*
\begin{proof}
  Let $G$ be the apex graph and $T$ a the spanning tree of $G$. If a part contains one of the $q$ apices, we give the entire tree $T$ to the part. This increases the congestion by at most $q$. For the remaining parts, we do the following. First, add an auxiliary new vertex $x$ that connects to each of the $q$ apices; the diameter can grow by at most $1$. Contract these $q+1$ vertices to a single apex to form graph $G'$; $T$ might now contain cycles, so take a spanning subtree of depth $d_T$ in the contracted $T$. Apply \Cref{lemma:GenusVortexShortcut} to the single apex graph $G'$. If we extend the shortcuts for each part in the natural way to $G$, the congestion does not change any further. Furthermore, the block parameter increases by at most $q-1$ because a block component containing $x$ splits into at most $q$ block components.
\end{proof}
}

\subsection{Conclusion and Open Problems}
 
\fullOnly{We conclude with an extended version and a proof of the main theorem.}
\shortOnly{We conclude with an extended version of the main theorem.}

\mainMinorLong*

\fullOnly{
\BP
We conclude with the proof of~\Cref{thm:mainMinorLong}. For a given excluded minor family $\m F$ of graphs, by~\Cref{thm:GraphStructure}, there exists a constant $k$ such that every graph $G \in \m F$ is also in $\m L_k$. By~\Cref{thm:almostEmbeddable}, for every $(k,k,k,k)$-almost-embeddable graph $H$, and for every spanning tree $T \s H$, there is a shortcut of block parameter $O(d_T)$ and congestion $O(d_T\log n)$; here, we treat $k$ as a constant. Plugging in $b:=O(d_T)$ and $c:=d_T\log n$ into~\Cref{thm:CliqueSum}, we conclude that every $k$-clique-sum of such graphs, and hence every graph in $\m L_k$, has shortcuts of block parameter $O(d_T)$ and congestion $O(d_T\log n + \log^2 n)$. The theorem follows from the definition of quality.
\EP
}

An obvious open question is whether the block parameter $O(d_T)$ can be improved to $\tilde O(1)$, which would result in a near-optimal $\tilde O(D)$-round algorithm for MST and $(1+\e)$-approximate mincut on excluded minor network graphs. The bottleneck in the current proof lies in the treewidth argument when arguing about Genus+Vortex graph, which produces the $O(d_T)$ block parameter. This treewidth argument cannot be improved due to lower bounds on treewidth-$k$ graphs, as presented in~\cite{haeupler2016near}. Hence, an improvement on Genus+Vortex graphs requires a better understanding of vortices, beyond treating them as simply low-treewidth (or pathwidth) graphs.

\newpage
\bibliographystyle{alpha}
\bibliography{refs}

\fullOnly{
\appendix

\newpage

\section{Combinatorial Gate in Genus+Vortex graphs}

\label{sec:proof:gateInGenusVortex}

In this section we prove \Cref{lemma:GenusVortexCombinatorial}. Namely, that the Genus+Vortex graphs have an $s$-combinatorial gate. The exposition is split into three parts: (1) a general ``Planarization'' Lemma that converts a genus-$g$ graph into a planar one, (2) proof that Genus-$g$ graphs have an $s$-combinatorial boundary, and finally (3) the proof of \Cref{lemma:GenusVortexCombinatorial}.

\subsection{Planarization of Genus-$g$ Graphs}

Genus-$g$ graphs are relatively hard to analyze. This is in contrast to planar graphs for which many tools are available. In this section we state a lemma that will help us analyzing genus-$g$ graphs by ``cutting and developing them on a plane''. The high-level idea is to ``cut'' the genus-$g$ graph $SG$ along multiple cycles, providing us with a ``planarization'' $PG$ that is planar (see \Cref{fig:torus-and-extension}). By ``cutting'' we informally mean taking scissors and cutting along edges in a cycle in a way that splits each edge into two sub-edges, one for each side of the cut. This also splits the nodes into multiple sub-nodes, possibly more than 2. For example, if a node lies on $k$ edge-disjoint cutting paths, we split the node into $2k$ sub-nodes. Such a planarization is illustrated on \Cref{fig:torus-and-extension}.

\begin{figure}[H]
  \centering
  \begin{subfigure}{0.40\textwidth} 
    \centering
    \includegraphics[width=\textwidth]{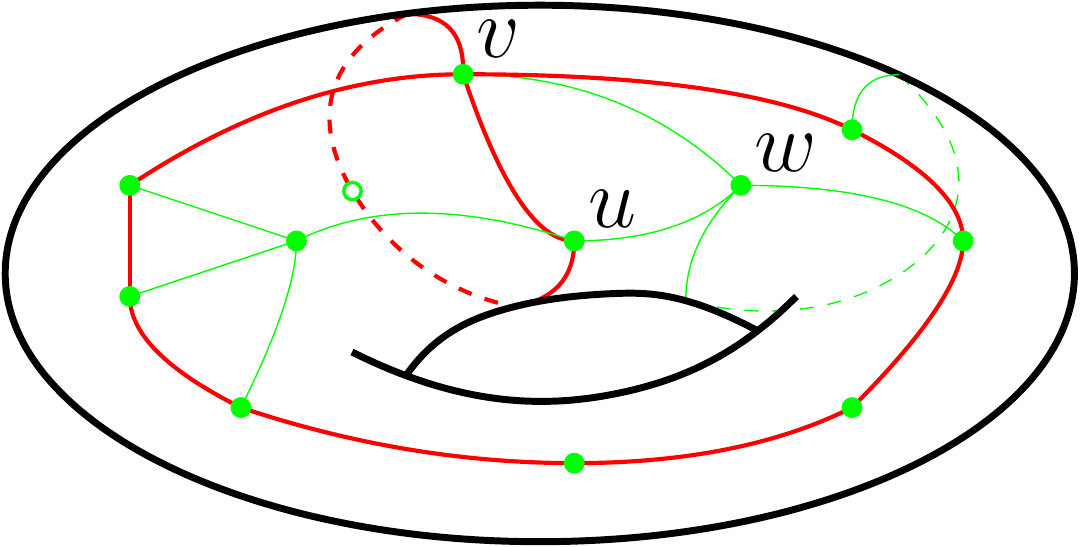}
  \end{subfigure}
  \begin{subfigure}{0.55\textwidth} 
    \centering
    \includegraphics[width=\textwidth]{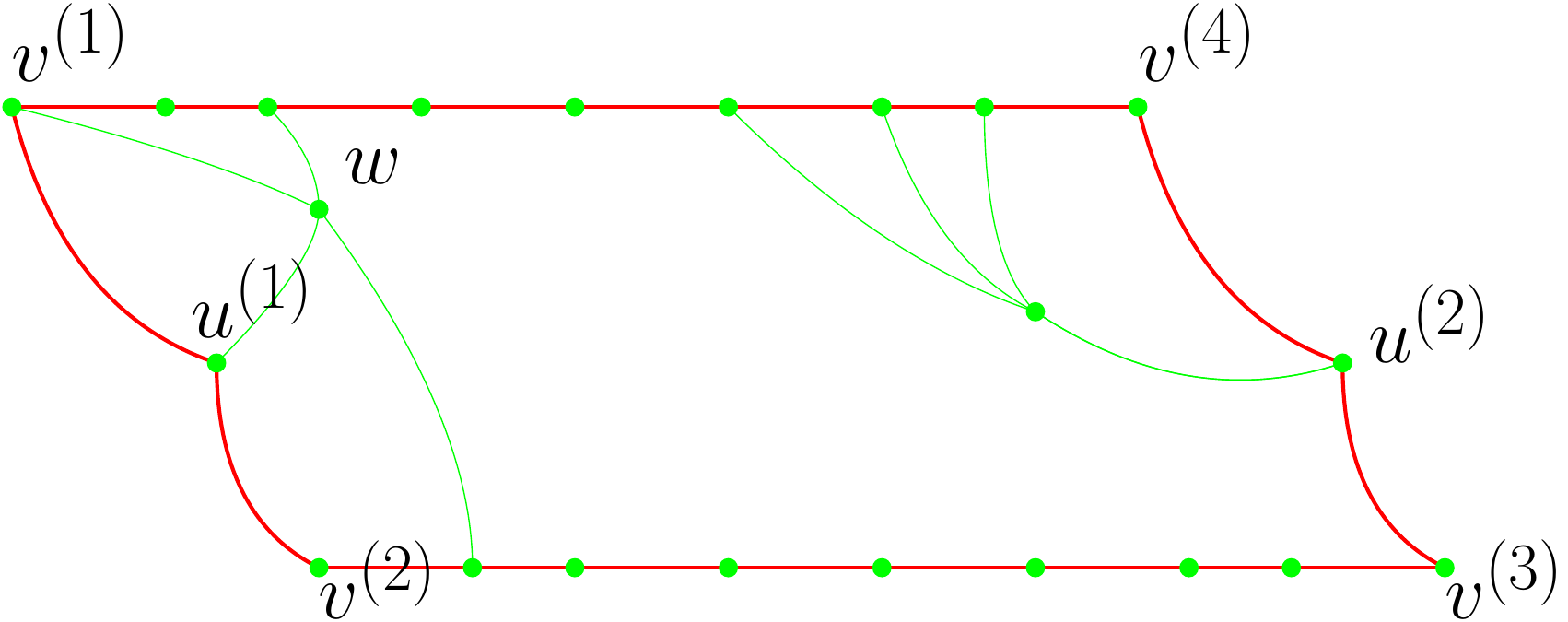}
  \end{subfigure}
  \caption{A graph embedded on a torus and its planarization after cutting the generators colored in red. Note how the vertex $v$ gets duplicated into $v^{(1)}, \ldots, v^{(4)}$.}
  \label{fig:torus-and-extension}
\end{figure}

On a more technical note, we have some control over the cycles which are cut. We can choose any spanning tree $T \subseteq SG$. Then the cut can be represented by $g$ non-tree edges and their induced cycles w.r.t. $T$. A cycle induced by a non-tree edge $e$ is the unique cycle in $T \cup \{e\}$. We first formalize the cutting procedure and then we state the Planarization Lemma in all all detail.

\begin{definition}[Cut Graph]
  \label{def:cutgraph}
  Given a graph $SG$ embedded on a surface and a subset of cut edges $R \subseteq E(SG)$, we define the \textbf{cut graph} $PG$ as follows.
  \begin{itemize}  \item For each $v \in V(SG)$ consider the local planar embedding on the surface (effectively an ordering of edges incident to $v$). Consider all maximal edge-intervals of non-cut edges in this ordering of $v$ including the two bounding cut edges. For all such edge-intervals construct a new copy of $v$ and include it in $V(PG)$. E.g., if $v$ is adjacent to $k$ cut edges, the vertex will be copied $k$ times.
  \item Define the "projection" $p : V(PG) \to V(SG)$ that maps $v \in V(PG)$ to the original vertex in $V(SG)$ of which $v$ is a copy.
  \item Connect two vertices in $V(PG)$ when their corresponding edge-intervals contain either (i) the same non-cut edge or (ii) the same cut edge while the edge-intervals are on the appropriate sides of their local embeddings. In other words, the cut edge is the clockwise boundary of one edge-interval and counter-clockwise boundary of the other; or vice versa. This should intuitively correspond to ``cutting an edge with scissors''.
  \end{itemize}

  On such a cut graph, denote by \textbf{outer nodes} all nodes $v \in V(PG)$ s.t. $|p^{-1}(p(v))| > 1$ and also their projections into $SG$. The outer nodes in $v \in V(SG)$ are exactly those corresponding to edge-intervals that do not contain the entire neighborhood, or equivalently, those without any cut edges incident to them.

  Similarly, denote the complement of outer nodes (on both $PG$ and $SG$) as the \textbf{inner nodes}. Note that, by the cutting procedure, there is a one-to-one correspondence between inner nodes on $PG$ and inner nodes on $SG$.
\end{definition}

\begin{lemma}[Planarization Lemma]
  \label{lemma:planarization}
  Consider a genus-$g$ graph $SG$ and a spanning tree $T \subseteq SG$. There exists a set of $g$ cycles induced by (adversarially chosen) non-$T$ edges. Call them the \textbf{generating cycles}. Let $PG$ be the cut graph of $SG$ with respect to the union of the generating cycles. Then (i) $PG$ is a planar graph, and (ii) outer nodes, as defined in \Cref{def:cutgraph}, are on the outer face of $PG$. 
\end{lemma}
\begin{proof}
  Consider the dual graph of $SG$, denoted by $SG^*$. The vertices/edges/faces of $SG^*$ correspond one-to-one with faces/edges/vertices of $SG$, respectively. Then there exists a spanning tree $T^*$ of $SG^*$, called a co-tree, that is disjoint from $T$. In other words, there is no edge of $T^*$ that corresponds to an edge of $T$. This claim is a direct consequence of Lemma 1 in Eppstein~\cite{eppstein2003dynamic}.

  Let $C$ be the set of edges not in either $T$ nor $T^*$. Then Lemma 2 in Eppstein~\cite{eppstein2003dynamic} asserts that the cycles induced by the edges of $T$ generate the fundamental group of the surface on which $SG$ is embedded. In other words, contracting the tree $T$ to a single node and deleting $E(T^*)$ would give us an embedded graph with exactly one node and one face. Denote this face by $F$. Because $SG$ is 2-cell embedded in the surface, $F$ is homeomorphic to some plane which we denote by $\Pi$.

  We are now ready to prove the Lemma. Our argument will be that the cut graph can be embedded without intersections in the plane $\Pi$. Note that the points that are not on the cut edges have a natural embedding into $\Pi$, namely, they have a position in $F$ which corresponds to a point in the plane $\Pi$. This fixes the embedding of the inner nodes of $PG$. However, this also forces the embedding of the outer nodes $v \in V(PG)$. Such nodes correspond to an edge-interval of $w \in V(SG)$. Since the surface points sufficiently close inside the edge-interval and arbitrarily close to $v$ are embedded into $\Pi$, we just set the embedding of $w$ to be the limit of such surface points.

  Such an embedding has no intersections since non-cut-edges correspond one-to-one with edges on $F$ which do not intersect. Cut edges cannot intersect by the cutting procedure. This proves claim (i). Also, note that the boundary of the face $F$ corresponds to the outer face boundary of $\Pi$ and that exactly outer nodes get mapped to that boundary. This proves claim (ii).
\end{proof}

\subsection{Combinatorial Gate in Genus-$g$ graphs}

We begin with a strengthening of \Cref{lemma:PlanarCombinatorial}, whose proof is immediate following the proof of \Cref{lemma:PlanarCombinatorial}. Then, we present the combinatorial gate proof for bounded genus graphs.

\BL\label{lemma:PlanarCombinatorial2}
Let $G$ be a planar graph with a cell partition of diameter $d$.
Fix a planar embedding of $G$. Then, there is an $36d$-combinatorial gate $\m S$, and furthermore, each node on the outer face of $G$ contained in a gate $S_i$ is also in the corresponding fence $F_i$.
\EL

\begin{lemma}
  \label{lemma:genus-comb-gate}
        A genus-$g$ graph $SG$ with a diameter-$d$ cell partition $\mathcal{C}_{SG}$ has a $O((g+1)d)$-combinatorial gate.
\end{lemma}
\begin{proof}
  The main idea is to planarize the graph using \Cref{lemma:planarization}, find a combinatorial gate for the planar graph and project it back to the surface graph $SG$. The details follow.

  We construct a rooted spanning tree $T$ of $SG$ by first constructing a spanning tree of each cell, connecting them arbitrarily into a spanning tree, and arbitrarily rooting it. Denote the inter-cell tree edges as \textbf{connecting edges}. For concreteness, assume directed tree edges go towards the root.

  Planarize the graph $SG$ into $PG = (V(PG), E(PG))$ w.r.t. $T$ using \Cref{lemma:planarization}. Let $p : V(PG) \to V(SG)$ denote the corresponding projection. We define a cell decomposition $\m C_{PG}$ that will intuitively match the decomposition on $\m C_{SG}$, except that some cells get split in multiple ones to respect the planarization. We formalize it by defining $\m C_{PG}$ on the planar graph $PG$ in an implicit manner: we will define them as connected components of a forest of rooted trees $T_{PG} \subseteq PG$ that is defined as follows. For $x, y \in V(PG)$ there is an directed edge $x \to y$ in $T_{PG}$ when all three of (i) $p(x)$ and $p(y)$ are in the same cell of $\mathcal{C}_{SG}$, (ii) $(p(x) \to p(y)) \in T$, and (iii) $\{ x, y \} \in E(PG)$ hold. This completely defines $T_{PG}$ and therefore $\m C_{PG}$.

  Note that (i) $T_{PG}$ is a tree (when ignoring directions) and (ii) each vertex $x \in V(PG)$ has outdegree at most 1. Claim (i) follows because any cycle in $T_{PG}$ would project into a cycle in $T$; and claim (ii) follows from construction since $p(x) \in T$ has outdegree at most 1 and the projections of neighboring nodes of $x$ in $PG$ are all distinct.

  Next, we analyze the cells $\m C_{PG}$. Note that cells in $\m C_{PG}$ have diameter at most $2d = O(d)$ since traveling via out-edges towards the root will reach it within $d$ steps. This can be argued from the projection of such a travel reaching its root in at most $d$ steps.

  Furthermore, we claim that each cell splits into at most $O(g+1)$ new cells. We can represent each cell $C \in \m C_{PG}$ with its root node in $T_{PG}$. And since every node in $PG$ corresponds to an edge-interval of a vertex in $SG$ (c.f. planarization), we can represent $C$ by an edge-interval. Finally, we say that $C \in \m C_{PG}$ splits from $C' \in \m C_{SG}$ when the projection of the root of $C$ maps to a node in $C'$.

  We fix a cell $C' \in \m C_{SG}$, i.e., on the surface graph, and argue about the number of nodes that split from from $C'$. Let $C \in \m C_{PG}$ be such a cell and let $v$ be the root of its corresponding component of $T_{PG}$. Note that node $v$ being a root in $T_{PG}$ implies that either (i) its edge-interval does not contain the unique outgoing edge out of $p(v)$ or (ii) it contains an outgoing connecting edge, i.e., unique edge connecting $T \subseteq SG$ with its parent. The number of cases (i) increases by $O(1)$ with each new generating cycle, of which there are $O(g+1)$. Case (ii) can occur only twice since the construction implies each edge gets duplicated at most twice. This concludes the argument that there are $|\m C_{PG}| \le O(g+1) |\m C_{SG}|$.

  To summarize, we have constructed a planar $PG$ and a diameter-$O(d)$ cell partition $\m C_{PG}$. We now apply \Cref{lemma:PlanarCombinatorial} to find a $O(d)$-combinatorial gate $\m S = \{ (F_i, S_i) \}_i$.

  In order to construct a combinatorial gate in $SG$, we project $\m S$. To that end, we extend the projection $p$ to work on subsets $2^{V(SG)}$ in the obvious manner: $p(A) = \bigcup_{a \in A} p(a)$. Next, let $\m S' = \{ p(F_i), p(S_i) \mid (F_i, S_i) \in \m S \}_i$. We claim that $\m S'$ is an $O((g+1)d)$-combinatorial gate in $SG$ and we prove it by verifying its properties one by one.

  (1) $p(F_i) \subseteq p(S_i)$ is clear because $p$ is an increasing function w.r.t. $\subseteq$.
  
  (2) We want to show that $\partial p(S_i) \subseteq p(F_i)$. Let $v' \in \partial p(S_i)$. If $v'$ is an inner node and $w' \in V(SG) \setminus p(S_i)$ is its neighbor outside of $p(S_i)$, then any preimage $w \in p^{-1}(w')$ must be outside of $S_i$, hence the unique $p^{-1}(v') \in \partial S_i \subseteq F_i$. This implies that $v' \in p(F_i)$ as needed. On the other hand, if $v'$ is not an inner node, then any preimage $v \in p^{-1}(v') \cap S_i$ must be an outer node, hence on the outer face of $PG$ by \Cref{lemma:PlanarCombinatorial}. If we go through the proof of \Cref{lemma:PlanarCombinatorial}, we observe that construction of the planar combinatorial gate in the lemma has an additional property: if a node on the outer face of the planar graph is contained in a gate $S_i$, then it is also in the corresponding fence $F_i$.  Thus by construction of $\mathcal{S}$ it is included in $F_i$, which implies $v' \in p(F_i)$.

  (3) Let $\{a', b'\} \in E(SG)$ be an edge whose endpoints are in different cells of $\m C_{SG}$. Then $\{ p(a'), p(b') \}$ is covered by a gate $S$ in $PG$. Hence the gate $p(S)$ covers $\{a', b'\}$.
  
  (4) Each gate $S'$ intersects at most 2 cells in $\mathcal{C}_{OG}$ since the projection maps the same cell into the same cell.
  
  (5) Let $v' \in p(S_i) \setminus p(F_i)$. We want to show that there can be at most one such $i$. On one hand, if $v' \in V(SG)$ is an outer node, then its preimage $p^{-1}(v') \cap S_i$ must be an outer node and hence included in $F_i$. This is a contradiction since then $v' \not \in p(S_i) \setminus p(F_i)$. On the other hand, if $v' \in V(SG)$ is an inner node, then it has a unique preimage $v = p^{-1}(v')$. If $v' \in p(S_i) \setminus p(F_i)$, then $v \in S_i \setminus F_i$. Hence by claim (5) on $\mathcal{S}$ there can be at most one such $i$.

        (6) $\sum_{(F', S') \in \mathcal{S'}} |F'| \le \sum_{(F, S) \in \mathcal{S}} |p(F)| \le \sum_{(F, S) \in \mathcal{S}} |F| \le O(d) |\mathcal{C}_{PG}| \le O((g+1)d) |\mathcal{C}_{OG}|$
\end{proof}

\subsection{Finalizing the Proof}

Finally, we extend the combinatorial gate proof for bounded genus graphs to include vortices.

\GenusVortexCombinatorial*

\begin{proof}
  
For notation, rename the graph $G$ to $OG$ for ``original graph''.  We first replace the Genus-$g$+Vortex-$k$ $OG$ graph with a tightly related genus-$g$ graph $SG$ by the following method: for each vortex $A$ in $OG$, remove all the internal vortex nodes and replace it with a star node $s_A$ that is connected to all nodes in the cycle of the vortex. Name the final graph $SG$ and note that it is genus $g$ by construction.

Next, we construct a corresponding cell partition $\m C_{SG}$: all non-star nodes get included in the same cell as they were in $OG$, while the star node gets included in the cell of its vortex. Note that the preconditions ensure that all vortex nodes are in the same cell. Furthermore, the diameter of the cell partition $\m C_{SG}$ is at most $d+1 = O(d)$, hence we can apply \Cref{lemma:genus-comb-gate} on it and obtain a $O((g+1)d)$-combinatorial gate $\mc S$.

We now convert $\mc S$ into a $O((g+1)kd)$-combinatorial gate $\mc S'$ on $OG$, hence completing the theorem. First, define an \textbf{expansion} $\mc E : V(SG) \to 2^{V(OG)}$ in the following manner: if $v$ is a star node, then $\mc E(v) = \emptyset$; if $v$ is a node on a vortex cycle, then $\mc E$ contains the set of all internal vortex nodes of $OG$ whose arcs contain $v$ and $v$ itself; finally, if $v$ is neither a star nor vortex cycle node, then $\mc E(v) = \{ v \}$. Note that $\m E$ furnishes a one-to-one correspondence between non-star nodes of $SG$ and non-internal vortex nodes of $OG$.

This allows us to define $\mc S'$. First, extend the expansion to work on subsets $2^{V(SG)}$ in the obvious manner: $\mc E(A) = \bigcup_{a \in A} \mc E(a)$. Then set $\m S' := \{ (\mathcal{E}(F_i), \mathcal{E}(S_i) \mid (F_i, S_i) \in \mathcal{S} \}$. In other words, we expand the gates and fences in $\mc S$ to obtain $\m S'$.

        We prove that $\m S'$ is an $O((g+1)kd)$-combinatorial gate by verifying its properties one by one.
  
  \begin{enumerate}
  
  \item It is clear that $F'_i \subseteq S'_i$ for each $(F'_i, S'_i) \in \m S$ since $\mathcal{E}$ is an increasing function w.r.t. $\subseteq$.

  \item We want to show that $\partial S'_i \subseteq F'_i$. Let $v' \in \partial S'_i$ where $S'_i = \m E(S_i), F'_i = \m E(F_i)$. We split into three cases, depending on the location of $v'$.
    \begin{itemize}
    \item If $v'$ is an internal vortex node, it is enough to prove that $F_i$ contains at least one node in the arc of $v'$, since this would imply that $v' \in F'_i$. Suppose, otherwise, that $F_i$ does not intersect the arc of $v'$. Knowing that $S_i$ contains at least one node in the arc of $v'$, we conclude that $S_i$ contains the entire arc of $v'$. Therefore, all neighbors of $v'$ on the vortex boundary are in $S'_i$. Moreover, every neighbor of $v'$ internal to the vortex must share a boundary vertex in the arc of $v'$, so the neighbor is also in $S'_i$. Therefore, all neighbors of $v'$ are contained in $S'_i$, a contradiction.
    \item If $v'$ is neither an internal vortex node nor on the vortex cycle, then there is a one-to-one correspondence via $\mathcal{E}$ between its neighborhood and gate incidence as in $SG$ since $v'$ cannot be connected to a vortex internal vertex. Hence the claim follows from the same claim for $\m S'$.
    \item If $v' \in \partial S'_i$ is on a vortex cycle, let $w' \in N(v') \bs S'_i$ be its neighbor outside of $S'_i$. If $w'$ is not internal to the vortex, then this case is equivalent to the previous one. If $w'$ is internal to the vortex, then $w'$ would be in $S'_i$ since $S'_i = \m E(S_i) \subseteq \m E(\m E^{-1}(v')) \ni w'$.
    \end{itemize}

  \item An edge $\{a', b'\} \in E(OG)$ whose endpoint are in different cells of $\m C_{OG}$ cannot have any of its endpoints as an internal vortex node, hence it has a corresponding edge $\{ \mathcal{E}^{-1}(a), \mathcal{E}^{-1}(b)\}$ in $E(SG)$. The claim now follows from the same claim in $\mathcal{S}$.

  \item Let $S'$ be a gate we want to prove intersects at most 2 cells. If $S'$ does not intersect the vortex internals, then it has a corresponding gate $\mathcal{E}^{-1}(S')$ in $\mathcal{S}$ from which the claim follows. If it intersects the vortex, then each internal vortex node belongs to the same cell denoted by $c$. Furthermore, the cells intersecting $S'$ are exactly the same as those intersecting $\mathcal{E}^{-1}(S')$ together with $c$. But the cells intersecting $\mathcal{E}^{-1}(S')$ are either an empty set or already include $c$, hence the claim follows.
      
  \item Let $v' \in S'_i \bs F'_i$. If $v'$ is not an internal vortex node, then there can be only one such $S' \bs F' \ni v'$ from the claim for $\mathcal{S}$. On the other hand, if $v'$ is a internal vortex node, then $S'_i \bs F'_i$ must contain the entire arc of $v'$. Then the claim follows from the non-internal vortex node case by picking any such node on the arc.

  \item We first note that $|\mathcal{E}(X)| \le k |X|$ since the vortex is of depth $k$. Then we have $\sum_{(F', S') \in \mathcal{S'}} |F'| \le \sum_{(F, S) \in \mathcal{S}} |\mathcal{E}(F)| \le k \sum_{(F, S) \in \mathcal{S}} |F| \le k O((g+1)d) |\mathcal{C}_{SG}| \le O((g+1)kd) |\mathcal{C}_{OG}|$.
  \end{enumerate}
\end{proof}
}
\end{document}